\documentclass[12pt, draftclsnofoot,peerreviewca,onecolumn]{IEEEtran}
\usepackage{multicol}
\usepackage{etoolbox}
\makeatletter
\patchcmd{\@makecaption}
  {\scshape}
  {}
  {}
  {}
\makeatletter
\patchcmd{\@makecaption}
  {\\}
  {.\ }
  {}
  {}
\makeatother

\usepackage{amsfonts}
\usepackage{mathrsfs}
\usepackage{mathtools}
\usepackage{amsfonts}
\usepackage{amssymb}
\usepackage{graphicx}
\usepackage{epsfig}
\usepackage{psfrag}
\usepackage{amsmath}
\usepackage{array}
\usepackage{cases}
\usepackage{eufrak}
\usepackage{cite,graphicx,amsmath,amssymb,color}
\usepackage{algorithmic}
\usepackage{algorithm}
\usepackage{bm}
\usepackage{multirow}
\usepackage{threeparttable}
\usepackage{array}
\usepackage{makecell}

\usepackage{color}
\definecolor{brown}{rgb}{0.85,0.33,0.1}
\definecolor{green}{rgb}{0.47,0.6,0.19}
\definecolor{purple}{rgb}{0.49,0.18,0.56}
\newcommand{\wqb}{\textcolor{black}}

\newcommand{\wqg}{\textcolor{black}}
\newcommand{\wqm}{\textcolor{black}}

\newtheorem{Thm}{Theorem}
\newtheorem{Lem}{Lemma}

\newtheorem{Prob}{Problem}

\usepackage{bbm}
\usepackage{amsfonts}
\usepackage{mathrsfs}
\usepackage{amsfonts}
\usepackage{amssymb}
\usepackage{graphicx}
\usepackage{epsfig}
\usepackage{epstopdf}
\usepackage{psfrag}
\usepackage{amsmath}
\DeclareRobustCommand{\stirling}{\genfrac\{\}{0pt}{}}

\usepackage{array}
\usepackage{subfigure}
\usepackage{cite,graphicx,amsmath,amssymb,color}
\usepackage{algorithmic}
\usepackage{algorithm}
\usepackage{algorithm}
\usepackage{algorithmic}
\usepackage{stmaryrd}
\usepackage{multirow}
\usepackage{subfig}
\usepackage{graphicx,times,amsmath} 
\usepackage{url}

\IEEEoverridecommandlockouts

\begin{document}
\bibliographystyle{IEEEtran}

\setcounter{page}{1}
\title{Optimization-based Decentralized Coded Caching for Files and Caches with Arbitrary Sizes}
\author{Qi Wang,
		Ying Cui,
		Sian Jin,
		Junni Zou,
		Chenglin Li,
		Hongkai Xiong
\thanks{Q. Wang, \ Y. Cui, \ S. Jin, \  J. Zou,  \ C. Li, \ H. Xiong are with Shanghai Jiao Tong University, China.}
}

\maketitle
\thispagestyle{headings}

\begin{abstract}
Existing decentralized coded caching solutions cannot guarantee small loads in the general scenario with arbitrary file sizes and cache sizes.
In this paper, we propose an optimization framework for decentralized coded caching \wqb{in the general scenario} to minimize the worst-case load and average load (under an arbitrary file popularity), respectively.
Specifically, we first propose a class of decentralized coded caching schemes \wqb{for the general scenario,} which are specified by a general caching parameter and include several known schemes as special cases.
Then, we optimize the caching parameter to minimize the worst-case load and average load, respectively.
\wqb{Each of the two optimization problems} is a challenging nonconvex problem with a nondifferentiable objective function.
\wqb{For each optimization problem}, we \wqb{develop} an iterative algorithm to obtain a stationary point using techniques for solving Complementary Geometric Programming (GP).
We also obtain a low-complexity \wqb{approximate} solution by solving an approximate problem with a differentiable objective function which is an upper bound \wqb{on} the original nondifferentiable one, and characterize the performance loss caused by the approximation. 
\wqb{Finally,} we present two information-theoretic converse bounds on the worst-case load and average load (under an arbitrary file popularity) in the general scenario, respectively.
To the best of our knowledge, this is the first work \wqb{that provides optimization-based decentralized coded caching schemes and information-theoretic converse bounds for the general scenario.}
\end{abstract}

\newpage
\section{Introduction}\label{Sec:Introduction}
%
%
%
%
%
Recently, a new class of caching schemes for content placement in user caches, referred to as {\em coded caching}~\cite{AliFundamental},  have received significant interest.
In~\cite{AliFundamental}, Maddah-Ali and Niesen consider a system consisting of one server with access to a library of files and connected  through a shared, error-free link to multiple users each with a cache.
Each user can obtain the requested file based on the received multicast message and the contents stored in its cache.
\wqg{They formulate a caching problem, consisting of two phases, i.e., uncoded content placement and coded content delivery, which has been successfully investigated in a large number of recent works\wqb{~\cite{AliFundamental,jin,Wei_Yu,AliDec, NonuniformDemands,ji2015order,zhang2018coded,Sinong,wang2015coded,zhang2015coded,cheng2017optimal}} under the same network setting.}

\wqb{In~\cite{AliFundamental,jin,Wei_Yu}, the authors propose centralized coded caching schemes to minimize the worst-case load (over all possible requests)~\cite{AliFundamental} or average load (over random requests)~\cite{jin,Wei_Yu} of the shared link in the delivery phase.
Specifically,} in~\cite{AliFundamental}, Maddah-Ali and Niesen propose a centralized coded caching scheme to minimize the worst-case load \wqg{and \wqb{show that it} achieves order-optimal memory-load tradeoff}.
In~\cite{jin}, Jin {\em et al.} consider a class of centralized  coded caching schemes specified by a  general file partition parameter, and optimize the parameter to minimize the average load  within the class under an arbitrary file popularity.
\wqg{In~\cite{Wei_Yu}, \wqb{a} parameter-based coded caching design approach \wqb{similar to the one} in~\cite{jin} is \wqb{adopted} to minimize the average load under an arbitrary file popularity \wqb{in the scenario with arbitrary file sizes and two different cache sizes}.}
Centralized coded caching schemes  have limited practical applicability, as they require a centrally coordinated placement phase depending on the \wqb{exact} number of active users in the delivery phase, which is \wqb{actually} not known \wqb{when placing content in a practical network}.

Decentralized coded caching schemes\wqm{\cite{AliDec,wang2015coded,zhang2015coded,cheng2017optimal,NonuniformDemands, ji2015order, zhang2018coded, Sinong}}, \wqb{where} the \wqb{exact} number of active users in the delivery phase \wqb{is not required and the cache of each user is filled independently of the other users}, are then considered \wqb{to minimize the worst-case load~\cite{AliDec,wang2015coded,zhang2015coded,cheng2017optimal} or average load~\cite{NonuniformDemands, ji2015order, zhang2018coded, Sinong}}.
In~\cite{AliDec}, Maddah-Ali and Niesen \wqb{propose a decentralized coded caching scheme to minimize the worst-case load and show that it \wqg{achieves order-optimal memory-load trade-off.}} 
\wqg{In~\cite{NonuniformDemands}, files \wqb{are partitioned} into multiple groups and \wqb{the} decentralized coded caching scheme in~\cite{AliDec} is applied to each group \wqb{to reduce the average load}.
As coded-multicasting opportunities for files from different groups are not explored, \wqb{the resulting average load may not be desirable.}
In~\cite{ji2015order}, \wqb{the authors optimize the memory allocation for files \wqm{by minimizing} \wqg{an upper bound on the average load}.}
\wqb{As the optimization problem} is highly non-convex and not amenable to analysis, \cite{ji2015order} proposes a simpler suboptimal scheme, referred to as the RLFU-GCC scheme, where all files are partitioned into two groups and the $GCC_1$ procedure is applied to the group of popular files.
In~\cite{zhang2018coded}, inspired by the RLFU-GCC scheme in~\cite{ji2015order}, Zhang \emph{et al.} present a \wqb{decentralized coded caching} scheme, \wqg{where all files are partitioned into two groups and the delivery procedure of Maddah-Ali-Niesen's decentralized scheme is applied to the group of popular files}. 
In~\cite{Sinong}, Wang {\em et al.} formulate a   coded caching design problem  to minimize the average load by  optimizing  the  cache memory allocation for files.
The optimization problem is nonconvex, and a low-complexity \wqb{approximate} solution is obtained by solving an approximate convex problem \wqb{of the original problem}.}
Note that~\cite{AliDec, NonuniformDemands, ji2015order, zhang2018coded, Sinong} all consider the scenario where all files have the same file size and all users have the same cache size, and the resulting schemes may not achieve small loads \wqb{when} file sizes or cache sizes \wqb{are different}.
\wqb{In contrast}, \cite{wang2015coded} considers the scenario with arbitrary cache sizes and the same file size, while~\cite{zhang2015coded} and \cite{cheng2017optimal} consider the scenario with arbitrary file sizes and the same cache size.
\wqb{More specifically}, the decentralized coded caching schemes proposed in~\cite{wang2015coded} and~\cite{zhang2015coded} are not optimization based, and the decentralized coded caching scheme in~\cite{cheng2017optimal} is obtained by solving an approximate problem \wqb{of the worst-case load minimization problem} without any performance guarantee.
Thus, the existing solutions in~\cite{wang2015coded,zhang2015coded,cheng2017optimal} \wqb{may not} guarantee small loads in the general scenario with arbitrary file sizes and cache sizes.

\wqg{Besides achievable schemes, \wqb{information-theoretic converse bounds on the worst-case load~\cite{AliFundamental,wang2015coded,zhang2015coded,yu2018exact,wang2016new,sengupta2015improved,ghasemi2017improved,yu2018characterizing,Chien_Yi_Wang2018improved} and average load~\cite{NonuniformDemands,ji2015order,zhang2018coded,Sinong,jin,yu2018exact,wang2016new,yu2018characterizing,Chien_Yi_Wang2018improved} for coded caching are presented.}
The converse bounds \wqb{in~\cite{AliFundamental,wang2015coded,zhang2015coded, sengupta2015improved,NonuniformDemands,ji2015order,zhang2018coded,Sinong,yu2018exact,ghasemi2017improved,wang2016new,Chien_Yi_Wang2018improved, yu2018characterizing,jin}} can be classified into two classes, i.e., class i): bounds that are applicable only to uncoded placement and class ii): bounds that are applicable to any placement (including uncoded placement and coded placement).
The \wqb{converse} bound in~\cite{yu2018exact} belongs to class i) and is exactly tight for both the worst-case load and average load.
The \wqb{converse} bounds \wqb{derived} based on reduction from an arbitrary file popularity to the uniform file popularity~\cite{NonuniformDemands,ji2015order,zhang2018coded,jin}, cut-set~\cite{AliFundamental,wang2015coded,zhang2015coded,Sinong},\wqb{\cite{sengupta2015improved}}, \wqb{association with a combinatorial problem of optimally labeling the leaves of a directed tree~\cite{ghasemi2017improved}}, relation between a multi-user single request caching network and a single-user multi-request caching network~\cite{wang2016new}, and other information-theoretic approaches~\cite{Chien_Yi_Wang2018improved,yu2018characterizing}, belong to class ii).
\wqb{Note that} the \wqb{converse} bounds in~\cite{AliFundamental,wang2015coded, sengupta2015improved,NonuniformDemands,ji2015order,zhang2018coded,Sinong,yu2018exact,wang2016new,ghasemi2017improved,Chien_Yi_Wang2018improved, yu2018characterizing,jin} \wqb{are for the scenario with the same file size and the converse bounds in~\cite{AliFundamental,zhang2015coded, sengupta2015improved,NonuniformDemands,ji2015order,zhang2018coded,Sinong,yu2018exact,wang2016new,ghasemi2017improved,Chien_Yi_Wang2018improved, yu2018characterizing,jin} are for the scenario with the same cache size, and hence cannot bound the minimum worst-case load or average load in the general scenario with arbitrary file sizes and arbitrary cache sizes.}}

In this paper, we consider the general scenario with arbitrary file sizes and cache sizes, and propose an optimization framework for decentralized coded caching in the general scenario to minimize the worst-case load \wqb{and average load}.
We also present two information-theoretic converse bounds on the worst-case load and average load (under an arbitrary file popularity), respectively, \wqg{applicable to any placement} in the general scenario.
\wqb{To our knowledge, this is the first work that provides optimization-based decentralized coded caching schemes and information-theoretic converse bounds for the general scenario.}
\wqg{Our detailed contributions are summarized below.}

\begin{itemize}
\item 
We propose a class of decentralized coded caching schemes utilizing general uncoded placement and a specific coded delivery, which are specified by a general caching parameter.
The considered class of decentralized coded caching schemes include the schemes in~\cite{AliDec, NonuniformDemands, ji2015order, zhang2018coded, Sinong} \wqb{(designed for the scenario with the same file size and cache size)}, the scheme in~\cite{wang2015coded} \wqb{(designed for the scenario with arbitrary cache sizes and the same file size)} and the schemes in~\cite{zhang2015coded} and \cite{cheng2017optimal} \wqb{(designed for the scenario with arbitrary file sizes and the same cache size)} as special cases.
Then, we formulate two \wqb{parameter-based} coded caching design optimization problems over the considered class of schemes to minimize the worst-case load \wqb{and average load, respectively}, \wqg{by optimizing the caching parameter.}
\wqb{Each problem is a challenging nonconvex problem with a nondifferentiable objective function.} 
\item 
\wqb{For each optimization problem,} we \wqb{develop} an iterative algorithm to obtain a stationary point by equivalently transforming the nonconvex problem to a Complementary Geometric Programming (GP) and using techniques for Complementary GP.
We also obtain a low-complexity \wqb{approximate} solution \wqb{with performance guarantee}, by bounding the original nondifferentiable objective function with two differentiable functions and solving an approximate problem with the differentiable upper bound as the objective function. 
\item 
\wqb{We present two information-theoretic converse bounds on the worst-case load and average load (under an arbitrary file popularity), respectively, \wqg{which are applicable to any placement (\wqg{including uncoded placement and coded placement})} in the general scenario.}
\wqb{In the scenario with the same file size and the same cache size, the proposed converse bounds on the worst-case load and average load reduce to the \wqm{two bounds} in~\cite{Chien_Yi_Wang2018improved} and \wqb{the one} in~\cite{jin}, respectively.}
\item
Numerical results show that the proposed solutions \wqg{achieve significant gains} over the schemes in \cite{AliDec,wang2015coded,zhang2015coded,cheng2017optimal} and the schemes in \cite{NonuniformDemands, ji2015order, zhang2018coded, Sinong} in the general scenario in terms of the \wqb{worst-case load} and average load, \wqb{respectively}.
\wqb{These results highlight the importance of designing optimization-based decentralized coded caching schemes for the general scenario.}
\end{itemize}
\begin{figure}
\begin{center}
  {\resizebox{7.5cm}{!}{\includegraphics{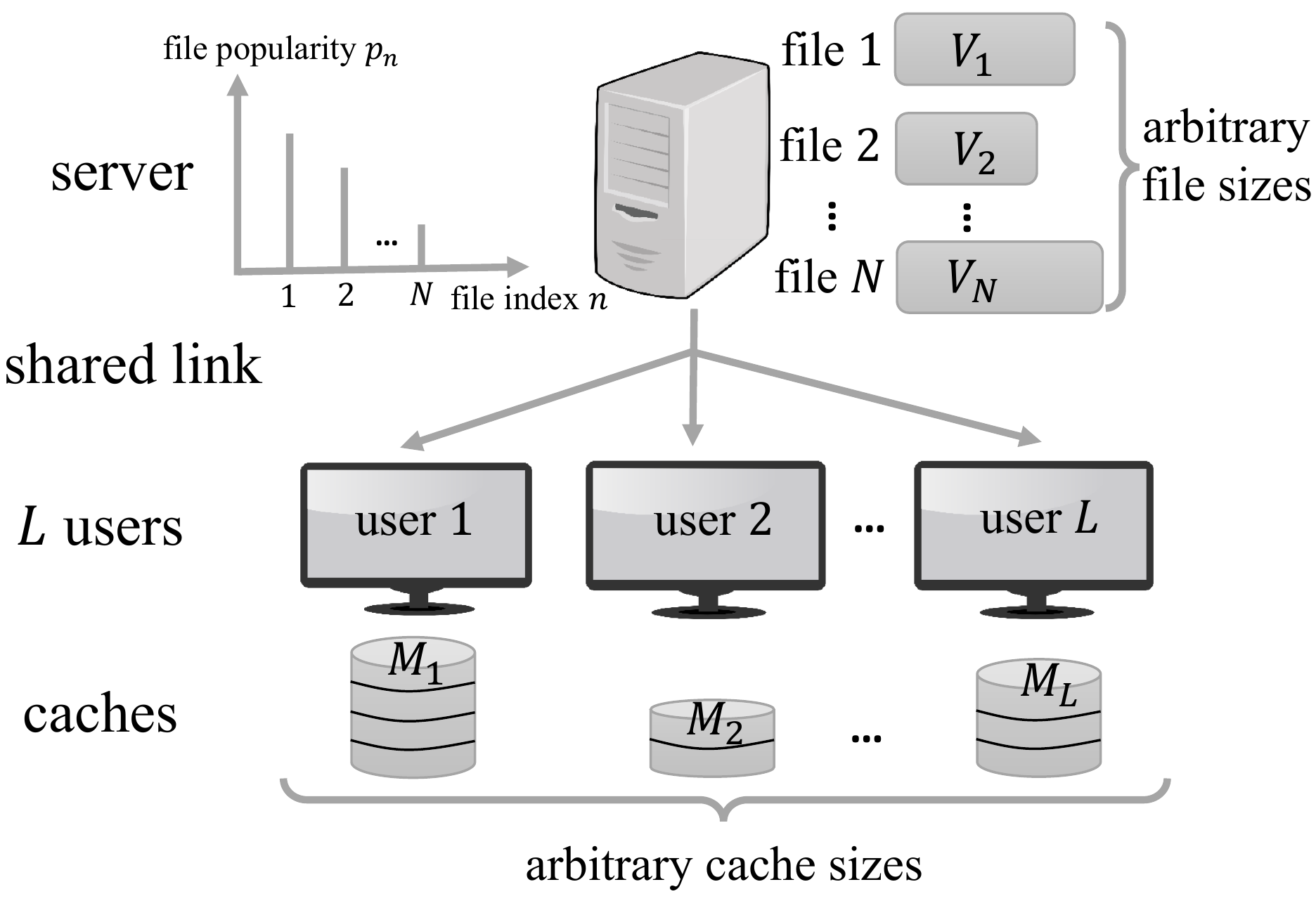}}}
         \caption{\small{\wqm{Problem setup in the general scenario with arbitrary file sizes and cache sizes.}}
         \vspace{-7mm}
         }\label{fig:system_model}
\end{center}
\end{figure}
\vspace{-2mm}
\section{Problem Setting}\label{Sec:Problem}
We consider a system with one server connected through a shared, error-free link to $L \in \mathbb N$  users \wqb{(including both active and inactive ones),} each with an isolated cache memory (see Fig.~\ref{fig:system_model})~\cite{AliDec}, where $\mathbb N$   denotes the set of all positive integers.\footnote{Note that the server can be a base station, and each user can be a mobile device or a small base station.}
Let $\mathcal{L} \triangleq \{1,2, \ldots L\}$ denote the set of user indices.
User $l\in \mathcal{L}$ has an isolated cache memory of size $M_l$ data units, where $M_l \in \left[0,\sum_{n=1}^N V_n\right].$
Let  $\mathbf{M}\triangleq (M_l)_{l \in \mathcal L}$ denote the cache sizes of all $L$ users.
Let $T$ denote the cardinality of set $\{M_l:l\in\mathcal{L}\}$.
That is, there are $T$ different values of the cache sizes for \wqb{the} $L$ users, denoted by $\overline{M}_1,\overline{M}_2,\dots,\overline{M}_T$.
Without loss of generality, we assume $\overline{M}_1< \overline{M}_2< \dots< \overline{M}_T$.
Denote $\mathcal{T}\triangleq \{1,2,\dots,T\}$.
\wqb{For all $t\in\mathcal{T}$,} let $\mathcal{L}_t\triangleq \big\{l\in\mathcal{L}:M_l=\overline{M}_t\big\}$ denote the set of user indices for the users whose cache sizes are $\overline{M}_t$.
\wqb{Let $L_t\triangleq |\mathcal{L}_t|$ denote the cardinality of set $\mathcal{L}_t$.}
The server has access to a library of $N \in \mathbb N$ files, denoted by $W_1,\dots,W_N$.
Let $\mathcal{N}\triangleq \{1,2,\ldots ,N\}$ denote the set of file indices. 
File $n\in\mathcal{N}$ consists of $V_n$ indivisible data units.
Let  $\mathbf{V}\triangleq (V_n)_{n \in \mathcal N}$ denote the file sizes of all $N$ files.
We assume that each user randomly and independently requests a file in $\mathcal{N}$ according to an arbitrary file popularity.
In particular, a user requests file $n$ with probability $p_n\in [0,1]$, where $n\in\mathcal{N}$.
Thus, the file popularity distribution is given by $\mathbf{p}\triangleq \left(p_n\right)_{n\in\mathcal{N}}$, where $\sum_{n=1}^N p_n=1$.
In addition, without loss of generality, we assume $p_1\ge p_2\ge \dots \ge p_N$.
Note that in this paper, we consider the general scenario with arbitrary file sizes and cache sizes ($V_n,n\in\mathcal{N}$ can be different and $M_l,l\in\mathcal{L}$ can be different, \wqb{i.e., $T\ge 1$}), which includes the scenarios considered in~\cite{AliFundamental,jin,AliDec,NonuniformDemands,ji2015order, zhang2018coded, Sinong} ($V_n,n\in\mathcal{N}$ are the same and $M_l,l\in\mathcal{L}$ are the same, i.e., $T=1$), \cite{wang2015coded} ($V_n,n\in\mathcal{N}$ are the same), \cite{zhang2015coded, cheng2017optimal} ($M_l,l\in\mathcal{L}$ are the same, i.e., $T=1$) and \cite{Wei_Yu} ($T=2$) as special cases.

The system operates in two phases, i.e., a placement phase and a delivery phase~\cite{AliDec}.
In the placement phase, \wqb{all} users in $\mathcal{L}$ are given access to the entire library of $N$ files.
Each user fills its cache by using the library of $N$ files in a decentralized manner (which will be illustrated in Section~\ref{Sec:Parameter-based}).
Let $\phi_l$ denote the caching function of user $l\in \mathcal{L}$, which maps the files $W_1, \ldots ,W_N$ into the cache content $Z_l\triangleq \phi_l(W_1, \ldots ,W_N)$ for user $l\in \mathcal{L}$, where $Z_l$ is of size $M_l$ data units.
Let $\pmb{\phi} \triangleq(\phi_1,\dots,\phi_L)$ denote the caching functions of all the $L$ users.
In the delivery phase, a subset of \wqb{$L_a$ users of $\mathcal{L}$} are active, denoted by $\mathcal{L}_a\subseteq\mathcal{L}$, and each active user in $\mathcal{L}_a$ randomly and independently requests one file in $\mathcal N$ according to the file popularity distribution $\mathbf{p}$.
\wqb{For all $t\in\mathcal{T}$,} let $\mathcal{L}_{a,t}\triangleq\big\{l\in\mathcal{L}_a:M_l=\overline{M}_t\big\}$ denote the set of user indices for the active users whose cache sizes are $\overline{M}_t$.
Let $L_{a,t}=|\mathcal{L}_{a,t}|$ denote the cardinality of \wqb{set} $\mathcal{L}_{a,t}$.
Thus, $L_a=\sum_{t\in\mathcal{T}} L_{a,t}$. 
Denote $\mathbf{L}_a\triangleq (L_{a,t})_{t\in\mathcal{T}}$.
Let $d_{l}\in\mathcal N$ denote the index of the file requested by \wqb{active} user $l \in \mathcal{L}_a $, and
let $\mathbf{d}_a\triangleq \left(d_l\right)_{l\in\mathcal{L}_a}\in \mathcal{N}^{L_a}$ denote the requests of  the $L_a$ active users.
The server replies to these $L_a$ requests by sending a multicast message over the shared link, observed by all $L_a$ active users.
Let $\psi_a$ denote the server encoding function, which maps the files $W_1, \ldots ,W_N$, cache contents $\mathbf{Z}_a\triangleq \big(Z_l\big)_{l\in\mathcal{L}_a}$ of the $L_a$ active users and requests $\mathbf{d}_a$ into the multicast message $Y\triangleq\psi_a(W_1, \ldots ,W_N,\mathbf{Z}_a,\mathbf{d}_a)$ sent by the server over the shared link.
Let $\mu_l$ denote the decoding function for \wqb{active} user $l\in\mathcal{L}_a$, which maps the multicast message $Y$ received over the shared link, the cache content $Z_l$ and the request $d_l$ to the estimate $\widehat{W}_{d_l}$ of the requested file $W_{d_l}$ of \wqb{active} user $l\in\mathcal{L}_a$.
Let $\pmb{\mu}_a\triangleq (\mu_l)_{l\in\mathcal{L}_a}$ denote the decoding functions of all the $L_a$ active users.
For a coded caching scheme defined by an encoding function $\psi_a$, caching functions $\pmb{\phi}_a$ and decoding functions $\pmb{\mu}_a$, the probability of error is defined as 
$$P_e(\psi_a,\pmb{\phi}_a,\pmb{\mu}_a) \triangleq  \max_{\mathbf{d}_a \in \mathcal{N}^{L_a}} \max_{l \in \mathcal{L}_a} \Pr\left[\widehat{W}_{d_l}\neq W_{d_l}\right]< \epsilon.$$
A coded  caching scheme is called {\em admissible} if  for every $\epsilon>0$ and every large enough $\mathbf{V}$ as in~\cite{AliDec,wang2015coded,zhang2015coded,cheng2017optimal,NonuniformDemands,ji2015order,zhang2018coded,Sinong}, $P_e(\psi_a,\pmb{\phi}_a,\pmb{\mu}_a)\le \epsilon$.
Since the shared  link is  error free, such admissible schemes exist for sure.
Given  an admissible coded caching scheme $\mathfrak{F}$ and the requests $\mathbf{d}_a$ of all \wqb{the $L_a$} active users, let $R(\mathcal{L}_a,N,\mathbf{V},\mathbf{M},\mathfrak{F},\mathbf{d}_a)$ be the length (expressed in data units) of the multicast message $Y$, where $R(\mathcal{L}_a,N,\mathbf{V},\mathbf{M},\mathfrak{F},\mathbf{d}_a)$ represents the load of the shared link.
Let
\begin{equation}
	R_{\rm wst}(\mathcal{L}_a,N,\mathbf{V},\mathbf{M},\mathfrak{F}) \triangleq \max_{\mathbf{d}_a \in \mathcal N^{L_a}} R(\mathcal{L}_a,N,\mathbf{V},\mathbf{M},\mathfrak{F},  \mathbf{d}_a)\nonumber
\end{equation}
denote the worst-case load  of the shared link.
Let
\begin{equation}
	R_{\rm avg}(\mathcal{L}_a,N,\mathbf{V},\mathbf{M},\mathfrak{F}) \triangleq \mathbb{E}\big[R(\mathcal{L}_a,N,\mathbf{V},\mathbf{M},\mathfrak{F},\mathbf{d}_a)\big]\nonumber
\end{equation}
denote the average load  of the shared link, where the average is taken over random requests.\footnote{\wqb{Later, we shall use different notations for the worst-case load and average load to reflect the dependency on the specific schemes considered.}}
Let
\begin{equation}
	R_{\rm wst}^*(\mathcal{L}_a,N,\mathbf{V},\mathbf{M}) \triangleq \min_{\mathfrak{F}} R(\mathcal{L}_a,N,\mathbf{V},\mathbf{M},\mathfrak{F})
\end{equation}
and
\begin{equation}
	R_{\rm avg}^*(\mathcal{L}_a,N,\mathbf{V},\mathbf{M}) \triangleq \min_{\mathfrak{F}}R(\mathcal{L}_a,N,\mathbf{V},\mathbf{M},\mathfrak{F})
\end{equation}
denote the minimum achievable worst-case load and the minimum achievable average load of the shared link, where the minimum is taken over all admissible coded caching schemes.
In this paper, we adopt \wqb{the coded} delivery strategy (i.e., encoding function $\psi_a$ and decoding functions $\pmb{\mu}_a$) \wqb{in~\cite{AliDec}}, and wish to \wqb{optimize the uncoded placement strategy (i.e., caching functions $\pmb\phi_a$) for decentralized coded caching to} minimize the worst-case load and average  load of the shared link in the delivery phase, in the general scenario with  arbitrary file sizes $\mathbf{V}$ and arbitrary cache sizes $\mathbf{M}$.


\section{Parameter-based Decentralized Coded Caching}\label{Sec:Parameter-based}
In this section, we first present a class of decentralized coded caching schemes \wqb{for} the general scenario \wqb{with arbitrary file sizes and cache sizes}, which are specified by a general caching parameter.
Then, we derive  the expressions of the worst-case load and average load as functions of the caching parameter for the class of schemes.

In the placement phase \wqb{(involving all $L$ users)}, for all $t\in\mathcal{T}$, each user $l\in\mathcal{L}_t$ independently caches a subset of $q_{t,n}V_n$ data units of file $n$, chosen uniformly at random, where
\begin{align}
&0 \leq q_{t,n} \leq 1 , \quad \forall t\in\mathcal{T}, \ n \in \mathcal{N},\label{eqn:q_range}\\
&\sum_{n=1}^{N} q_{t,n}V_n\le \overline{M}_t, \quad  \forall t\in\mathcal{T}. \label{eqn:memory_constraint}
\end{align}
Note that \eqref{eqn:memory_constraint} represents the  cache memory constraint.
The caching parameter $\mathbf{q} \triangleq \left(q_{t,n}\right)_{t\in\mathcal{T},n\in\mathcal{N}}$ is a design parameter and will be optimized subject to the constraints in \eqref{eqn:q_range} and \eqref{eqn:memory_constraint} to minimize the worst-case load and average load in Section~\ref{Sec:Worst-case_Load_Minimization} and Section~\ref{Sec:Average_Load_Minimization}, respectively.
The general uncoded placement strategy parameterized by $\mathbf{q}$ extends those in~\cite{AliDec} ($T=1$ and $q_{1,n},n\in\mathcal{N}$ are the same), \cite{wang2015coded} ($q_{t,n},n\in\mathcal{N}$ are the same, for all $t\in\mathcal{T}$) and\cite{zhang2015coded,cheng2017optimal,NonuniformDemands,ji2015order,zhang2018coded,Sinong} ($T=1$).
As in~\cite{AliDec,wang2015coded,zhang2015coded,cheng2017optimal,NonuniformDemands,ji2015order,zhang2018coded,Sinong}, the random uncoded placement procedure can be operated in a decentralized manner in the sense that the \wqb{exact} number of active users in the delivery phase is \wqb{not} required and  the cache of each user is filled independently of the other users.

The coded delivery procedure \wqb{(involving the $L_a$ active users)} is the same as those in \cite{AliDec,Sinong,wang2015coded,zhang2015coded,cheng2017optimal} and is briefly presented here for completeness.
Let $W_{d_l,\mathcal{S}}$ denote the data units of file $d_l$ requested by \wqb{active} user $l\in\mathcal{L}_a$ that are  cached exclusively at the \wqb{active} users in $\mathcal{S}\subseteq \mathcal{L}_a$ (i.e., every data unit in $W_{d_l,\mathcal{S}}$ is present in the cache of every \wqb{active} user in $\mathcal{S}$ and is absent from the cache of every \wqb{active} user outside $\mathcal{S}$).
For any $\mathcal{S}\subseteq \mathcal{L}_a$ of cardinality $|\mathcal{S}|=s\in\{1,2,\dots,L_a\}$, the server transmits coded-multicast message $\oplus_{l\in\mathcal{S}}W_{d_l,\mathcal{S}\setminus \{l\}}$, where operator $\oplus$ denotes componentwise XOR.
All elements in the coded-multicast message are assumed to be zero-padded to the length of the longest element.
For all $s\in\{1,2,\dots,L_a\}$, we conduct the above delivery procedure.
The multicast message \wqb{for the $L_a$ active users} is simply the concatenation of the coded-multicast messages for all $s\in\{1,2,\dots,L_a\}$.
By the proof of Theorem 1 in~\cite{AliDec}, we can conclude that each \wqb{active} user \wqb{$l\in\mathcal{L}_a$} is able to decode the requested file based on the received multicast message and its cache content.

Now we formally summarize the placement and delivery procedures of the class of the decentralized coded caching schemes specified by the general caching  parameter $\mathbf{q}$ in Algorithm~\ref{alg:alg1}.
\begin{algorithm}[t]
\begin{small}
\caption{Parameter-based Decentralized Coded Caching}
\textbf{Input}: $\mathbf{q}=(q_{t,n})_{t\in\mathcal{T},n\in\mathcal{N}}$

\textbf{Placement procedure}
\begin{algorithmic}[1]
\FOR {$t\in\mathcal{T}$}
	\FOR {$l \in \mathcal{L}_t$, $n \in \mathcal N$}
  		\STATE user $l$ independently caches a subset of $q_{t,n}V_n$ data units of file $n$, chosen uniformly at random
  	\ENDFOR
\ENDFOR
\end{algorithmic}
\textbf{Delivery procedure~\cite{AliDec}}
\begin{algorithmic}[1]
\FOR {$s\in\{1,2,\dots,L_a\}$}
  \FOR {$\mathcal{S} \subseteq \mathcal{L}_a:|\mathcal{S}|=s$}
     \STATE Server sends $\oplus_{l \in \mathcal{S}} W_{d_l,\mathcal{S}\setminus \{l\}}$
  \ENDFOR
\ENDFOR
\end{algorithmic}\label{alg:alg1}
\end{small}
\end{algorithm}
By Algorithm~\ref{alg:alg1} and \cite{AliDec}, we can calculate the worst-case load for the \wqb{$L_a$} active users in $\mathcal{L}_a$ under given caching parameter $\mathbf{q}$ as follows:
\begin{align}
\max_{\mathbf{d}_a\in \mathcal{N}^{L_a}}\sum_{s=1}^{L_a}\sum_{\mathcal{S}\subseteq\mathcal{L}_a:|\mathcal{S}|=s }\max_{l\in \mathcal{S}}
\left(\prod_{a\in \mathcal{S}\backslash\{l\}}\prod_{t:a\in\mathcal{L}_{a,t}} q_{t,d_l}\right)
\left(\prod_{b\in\mathcal{L}_a\backslash(\mathcal{S}\backslash\{l\})}\prod_{t:b\in\mathcal{L}_{a,t}}(1-q_{t,d_l})\right)V_{d_l}, \label{eqn:worst_load_0}
\end{align}
where 
$\underset{l\in \mathcal{S}}\max
\left(\underset{a\in \mathcal{S}\backslash\{l\}}\prod\underset{t:a\in\mathcal{L}_{a,t}}\prod q_{t,d_l}\right)
\left(\underset{b\in\mathcal{L}_a\backslash(\mathcal{S}\backslash\{l\})}\prod\underset{t:b\in\mathcal{L}_{a,t}}\prod (1-q_{t,d_l})\right)V_{d_l}$ represents  the length of the coded-multicast message $\oplus_{l \in \mathcal{S}} W_{d_l,\mathcal{S}\setminus \{l\}}$.
\wqb{It is clear that the expression in \eqref{eqn:worst_load_0} is a function of $\mathbf{L}_a=(L_{a,t})_{t\in\mathcal{T}}$, library size $N$, file sizes $\mathbf{V}$ and caching parameter $\mathbf{q}$, and can be rewritten as}
\begin{align}
&R_{\rm wst}(\mathbf{L}_a,N,\mathbf{V},\mathbf{q})\triangleq\nonumber\\
 &\max_{\mathbf{n}\in \mathcal{N}^{L_a}}\sum_{s=1}^{L_a}\sum_{\substack{\mathcal{S}\subseteq\{1,\dots,L_a\}:\\|\mathcal{S}|=s }}\max_{j\in \mathcal{S}}
\left(\prod_{a\in \mathcal{S}\backslash\{j\}}\prod_{\substack{t:a\in\mathcal{I}_t(\mathbf{L}_a)}}q_{t,n_j}\right)
\left(\prod_{\substack{b\in\{1,\dots,L_a\}\\ \backslash  (\mathcal{S}\backslash\{j\})}}\prod_{t:b\in\mathcal{I}_t(\mathbf{L}_a)}(1-q_{t,n_j})\right)V_{n_j}, \label{eqn:worst_load_1}
\end{align}
\wqb{where $L_a=\sum_{t\in\mathcal{T}}L_{a,t}$, $I_t(\mathbf{L}_a)\triangleq \sum_{i=1}^t L_{a,i}$, $t\in\mathcal{T}$, $I_0(\mathbf{L}_a)\triangleq 0$, $\mathcal{I}_t(\mathbf{L}_a)\triangleq\big\{I_{t-1}(\mathbf{L}_a)+1,I_{t-1}(\mathbf{L}_a)+2,\dots, I_t(\mathbf{L}_a)\big\}$, $t\in\mathcal{T}$ and $\mathbf{n}\triangleq (n_1,n_2,\dots,n_{L_a})$.}
\wqb{Similarly, we can calculate the average load for the $L_a$ active users in $\mathcal{L}_a$ under given caching parameter $\mathbf{q}$ and rewrite it as:}
\begin{align}
&R_{\rm avg}(\mathbf{L}_a,N,\mathbf{V},\mathbf{q}) \triangleq \nonumber\\
&\sum_{\mathbf{n}\in \mathcal{N}^{L_a}}\left(\prod_{j=1}^{L_a} p_{n_j}\right)\sum_{s=1}^{L_a}\sum_{\substack{\mathcal{S}\subseteq\{1,\dots,L_a\}:\\|\mathcal{S}|=s }}\max_{j\in \mathcal{S}}
\left(\prod_{a\in \mathcal{S}\backslash\{j\}}\prod_{\substack{t:a\in\mathcal{I}_t(\mathbf{L}_a)}}q_{t,n_j}\right)
\left(\prod_{\substack{b\in\{1,\dots,L_a\}\\ \backslash(\mathcal{S}\backslash\{j\})}}\prod_{t:b\in\mathcal{I}_t(\mathbf{L}_a)}(1-q_{t,n_j})\right)V_{n_j}.\label{eqn:average_load_1}	
\end{align}


\wqb{We wish to optimize the caching parameter to minimize the worst-case load and average load of the shared link, both depending on $\mathbf{L}_a$.
As $\mathbf{L}_a$ is unknown in the placement phase in a decentralized setting, we consider a carefully chosen substitute $\mathbf{K}\triangleq(K_t)_{t\in\mathcal{T}}$ for $\mathbf{L}_a$ based on some prior information.
This parameter is referred to as the optimization parameter when formulating coded caching design optimization problems.
For example, optimization parameter $\mathbf{K}$ can be chosen as $K_t=\lceil \mathbb{E}[L_{a,t}]\rceil,t\in\mathcal{T}$, assuming that $\mathbf{L}_a$ is random and the distribution is known.}
\wqb{Please note that} the proposed optimization framework in this paper successfully extends \wqb{the parameter-based optimization framework} in our previous work~\cite{jin} which is for centralized coded caching in the scenario with the same file size and cache size.
\wqb{In addition, it is worth noting that the optimization-based frameworks for designing decentralized coded caching in\wqm{\cite{ji2015order,Sinong,cheng2017optimal}} rely on the exact number of active users in the delivery phase, and hence cannot be directly applied to a decentralized setup.
This problem can be appropriately handled by using a substitute for the exact number of active users in the delivery phase, as discussed above.}
\section{Worst-case Load Minimization}\label{Sec:Worst-case_Load_Minimization}
In this section, we first formulate a parameter-based coded caching design optimization problem  over the considered class of schemes to minimize the worst-case load by optimizing the caching parameter.
\wqb{Next}, we develop an iterative algorithm to obtain a stationary point of an equivalent problem.
\wqb{Finally}, we obtain an \wqb{\wqb{approximate} solution with low computational complexity and small worst-case load,} and characterize \wqb{its} performance loss.
To the best of our knowledge, this is the first work obtaining optimization-based decentralized coded caching design  to \wqb{reduce} the worst-case load in the general scenario with arbitrary file sizes and cache sizes.

\subsection{Problem Formulation}\label{Subsec:Problem_Formulation(Rwst)}

As the caching parameter  $\mathbf{q}$ fundamentally affects \wqb{the worst-case load, we would like to optimize $\mathbf{q}$ subject to the constraints in \eqref{eqn:q_range} and \eqref{eqn:memory_constraint}
so as to minimize $R_{\rm wst}(\mathbf{K},N,\mathbf{V},\mathbf{q})$.
}
\begin{Prob}[Caching Parameter Optimization \wqb{for Reducing} Worst-case Load]\label{P:Worst Case1}
\begin{align}
	R^*_{\rm wst}(\mathbf{K},N,\mathbf{V})\triangleq \min_{\mathbf{q}} \quad &  R_{\rm wst}(\mathbf{K},N,\mathbf{V},\mathbf{q})\nonumber \\
	\mathrm{s.t.} \quad  &\eqref{eqn:q_range}, \eqref{eqn:memory_constraint}, \nonumber
\end{align}
\end{Prob}
\wqb{where $R_{\rm wst}(\mathbf{K},N,\mathbf{V},\mathbf{q})$ is given in \eqref{eqn:worst_load_1}. }

As the objective function of Problem~\ref{P:Worst Case1} is nonconvex and nondifferentiable, Problem~\ref{P:Worst Case1} is a challenging problem.
A classical goal for solving a nonconvex problem is to obtain a stationary point.

\subsection{Solutions}\label{Subsec:Solutions(Rwst)}
\subsubsection{Stationary Point}\label{Subsubsec:Stationary Point(Rwst)}
First, we obtain an equivalent \wqb{nonconvex} problem of Problem~\ref{P:Worst Case1}.\footnote{Note that in this paper, $\succ$ and $\succeq$ represent componentwise inequalities.
In addition, note that for ease of analysis, we consider $\mathbf{q,x,w}\succ \mathbf{0},u>0$ instead of $\mathbf{q,x,w}\succeq \mathbf{0},u\ge0$, which does not change the optimal value or affect the numerical solution.}
\begin{Prob}[Equivalent Complementary GP of Problem~\ref{P:Worst Case1}]\label{P:Worst Case2}
\begin{align}
 &\min_{\mathbf{q},\mathbf{x,w}\succ \mathbf{0},u>0}  \quad u\nonumber\\
   &\ \quad\ \mathrm{s.t.} \ \quad\sum_{n=1}^N q_{t,n}V_n \overline{M}_t^{-1}\le 1,\quad \forall t\in\mathcal{T},\label{eqn:wstcgp_1} \\
 &\ \quad\ q_{t,n}\le 1,\quad \forall t\in\mathcal{T},\ n\in\mathcal{N},\label{eqn:wstcgp_2}\\
 &\ \quad\ \frac{1}{q_{t,n}+x_{t,n}}\le 1,\quad \forall t\in\mathcal{T},\ n\in\mathcal{N},\label{eqn:rational function}\\
 &\left(\prod_{a\in \mathcal{S}\backslash\{j\}}\prod_{\substack{t:a\in\mathcal{I}_t(\mathbf{K})}}q_{t,n_j}\right)
\left(\prod_{\substack{b\in\mathcal{K}\backslash  (\mathcal{S}\backslash\{j\})}}\prod_{t:b\in\mathcal{I}_t(\mathbf{K})}x_{t,n_j}\right)V_{n_j} w_{\mathbf{n},\mathcal{S}}^{-1}\le 1,\nonumber\\
 &\ \quad\ \forall \mathbf{n}\in\mathcal{N}^K,\ \mathcal{S}\subseteq\mathcal{K}:|\mathcal{S}|=s,\ j\in \mathcal{S},\label{eqn:wstcgp_4}\\
 &\ \quad\ \sum_{s=1}^K\sum_{\mathcal{S}\subseteq\mathcal{K}:|\mathcal{S}|=s }w_{\mathbf{n},\mathcal{S}} u^{-1}\le 1,\quad \forall\mathbf{n}\in\mathcal{N}^K\label{eqn:wstcgp_5},
\end{align}
\end{Prob}
where \wqb{$\mathbf{x}\triangleq \left(x_{t,n}\right)_{t\in\mathcal{T},n\in\mathcal{N}}$ , $\mathbf{w}\triangleq\left(w_{\mathbf{n},\mathcal{S}}\right)_{\mathbf{n}\in\mathcal{N}^K,\mathcal{S}\subseteq \mathcal{K}}$, $K\triangleq \sum_{t\in\mathcal{T}} K_t$ and $\mathcal{K}\triangleq \{1,2,\dots,K\}$.}
\vspace{1mm}
\begin{Lem}[Equivalence between Problem~\ref{P:Worst Case1} and Problem~\ref{P:Worst Case2}]\label{L:lem1} The optimal values of Problem~\ref{P:Worst Case1} and Problem~\ref{P:Worst Case2} are the same.
\end{Lem}
\begin{proof}
\wqb{Please refer to Appendix A.	}
\end{proof}

Note that Problem~\ref{P:Worst Case2} minimizes a posynomial subject to 
upper bound inequality constraints on posynomials (i.e., \eqref{eqn:wstcgp_1}, \eqref{eqn:wstcgp_2}, \eqref{eqn:wstcgp_4} \wqb{and} \eqref{eqn:wstcgp_5}) and upper bound inequality constraints on the ratio between two posynomials (i.e., \eqref{eqn:rational function}), and hence is a Complementary GP (which is not a conventional GP).
A stationary point of a Complementary GP can be obtained by solving a sequence of approximate GPs~\cite{chiang2007power}.
Specifically, at iteration $i$, $\left(\mathbf{q}^{(i)},\mathbf{x}^{(i)},\mathbf{w}^{(i)},u^{(i)}\right)$ is updated by solving the following approximate GP of Problem~\ref{P:Worst Case2}, which is parameterized by $\left(\mathbf{q}^{(i-1)},\mathbf{x}^{(i-1)},\mathbf{w}^{(i-1)},u^{(i-1)}\right)$ obtained at iteration $i-1$.
%
\begin{Prob}[Approximate GP \wqb{of Problem~\ref{P:Worst Case2}} at Iteration $i$]\label{P:Worst Case4}
 \begin{align}
 &\left (\mathbf{q}^{(i)},\mathbf{x}^{(i)},\mathbf{w}^{(i)},u^{(i)}\right)\triangleq \mathop{\arg\min}_{\mathbf{q},\mathbf{x,w}\succ \mathbf{0},u>0}  \quad  u\nonumber\\
   &\mathrm{s.t.}\quad \eqref{eqn:wstcgp_1},\eqref{eqn:wstcgp_2},\eqref{eqn:wstcgp_4},\eqref{eqn:wstcgp_5}, \nonumber \\
 &\frac{1}{\left(q_{t,n}^{(i-1)}+x_{t,n}^{(i-1)}\right)\left(\frac{q_{t,n}}{q_{t,n}^{(i-1)}}\right)^{\alpha_{t,n}^{(i-1)}}\left(\frac{x_{t,n}}{x_{t,n}^{(i-1)}}\right)^{\beta_{t,n}^{(i-1)}}}\le 1,
 \quad\forall t\in\mathcal{T},n\in\mathcal{N},\label{eqn:wstcgp_8}
 \end{align}
\end{Prob}
where
\begin{equation}
	\alpha_{t,n}^{(i-1)}\triangleq\frac{q_{t,n}^{(i-1)}}{q_{t,n}^{(i-1)}+x_{t,n}^{(i-1)}},\ \beta_{t,n}^{(i-1)}\triangleq\frac{x_{t,n}^{(i-1)}}{q_{t,n}^{(i-1)}+x_{t,n}^{(i-1)}}. \nonumber
\end{equation}

Problem \ref{P:Worst Case4} is a standard form GP, which can be readily converted into a convex problem and solved efficiently (e.g., using interior point methods) \cite{bertsekas1999nonlinear}. 
The details are summarized in Algorithm~\ref{alg:alg2}.
By \cite{marks1978general}, we know that  $(\mathbf{q}^{(i)},\mathbf{x}^{(i)},\mathbf{w}^{(i)},u^{(i)})$ is provably convergent to a stationary point of Problem~\ref{P:Worst Case2}, as $i\to \infty$. 
\begin{algorithm}[t]
    \caption{Algorithm for Obtaining a Stationary Point of Problem~\ref{P:Worst Case2}}
        \begin{algorithmic}[1]
           \STATE \textbf{Initialization}: choose any feasible point $\left(\mathbf{q}^{(0)},\mathbf{x}^{(0)},\mathbf{w}^{(0)},u^{(0)}\right)$ and set $i=1$\\
           \STATE \textbf{repeat}
           \STATE \quad Compute $\left(\mathbf{q}^{(i)},\mathbf{x}^{(i)},\mathbf{w}^{(i)},u^{(i)}\right)$ by solving Problem \ref{P:Worst Case4}  using interior point methods
           \STATE\quad Set $i=i+1$
           \STATE \textbf{until} some convergence criterion is met
    \end{algorithmic}
    \label{alg:worst}\label{alg:alg2}
\end{algorithm}
\subsubsection{Low-complexity \wqb{Approximate} Solution}\label{Subsubsec:Low-complexity_Solution(Rwst)}
As the numbers of variables and constraints of Problem~\ref{P:Worst Case4} \wqb{are $O\left( N^K \right)$ and $O\left( N^K K \right)$, respectively,} and  Problem~\ref{P:Worst Case4} is solved at each iteration, the \wqb{computational} complexity of Algorithm~\ref{alg:alg2} is high for large $N$ or $K$.
In this part, we aim at obtaining a \wqb{low-complexity approximate solution of Problem~\ref{P:Worst Case1} with small worst-case load}, which is applicable especially for large $N$ or $K$.

First, we approximate Problem~\ref{P:Worst Case1} with a nondifferentiable objective function by a problem with a differentiable objective function.
The nondifferentiable max function can be bounded from above and below by the following differentiable functions:
\begin{equation}
\max\{x_1,x_2,\dots,x_n\}\le\frac{\ln(e^{cx_1}+\dots+e^{cx_n})}{c},\label{ineqn:max_appro_ub}
\end{equation}
\begin{equation}
\max\{x_1,x_2,\dots,x_n\}\ge \frac{\ln(e^{cx_1}+\dots+e^{cx_n})}{c}-\frac{\ln(n)}{c},\label{ineqn:max_appro_lb}
\end{equation}
where $c\ge 1$.
It is clear that the upper and lower bounds are asymptotically tight as $c\to \infty $, and the equality in \eqref{ineqn:max_appro_lb} holds when $x_1=\dots=x_n$.
In addition, the upper bound in \eqref{ineqn:max_appro_ub} is tighter when the deviation of $x_1,\dots,x_n$ is larger, and the lower bound in \eqref{ineqn:max_appro_lb} is tighter when the deviation is smaller.
By \eqref{ineqn:max_appro_ub}, we can obtain an upper bound of \wqb{$R_{\rm wst}(\mathbf{K},N,\mathbf{V},\mathbf{q})$}, denoted by 
\wqb{
\begin{align}
R^{\rm up}_{\rm wst}(\mathbf{K},N,\mathbf{V},\mathbf{q})
\triangleq &\frac{1}{c}\ln\Bigg (\sum_{\mathbf{n}\in \mathcal{N}^K}\exp\Bigg(c\sum_{s=1}^K\sum_{\mathcal{S}\subseteq\mathcal{K}:|\mathcal{S}|=s}\frac{1}{c}\ln\Bigg(\sum_{j\in\mathcal{S}}\exp\Bigg(c\Bigg(\prod_{a\in \mathcal{S}\backslash\{j\}}\prod_{\substack{t:a\in\mathcal{I}_t(\mathbf{K})}}q_{t,n_j}\Bigg)\nonumber\\
&\times\Bigg(\prod_{\substack{b\in\mathcal{K}\backslash  (\mathcal{S}\backslash\{j\})}}\prod_{t:b\in\mathcal{I}_t(\mathbf{K})}\left(1-q_{t,n_j}\right)\Bigg)V_{n_j}\Bigg)\Bigg) \Bigg)\Bigg ).\label{ineq:relaxmodel_ub}
\end{align}
}
Therefore, we can approximate Problem~\ref{P:Worst Case1} by minimizing \wqb{the upper bound in \eqref{ineq:relaxmodel_ub}}.


\begin{Prob}[Approximate Problem of Problem~\ref{P:Worst Case1}]\label{P:Worst Case max}
\begin{align}
 \wqb{R^{\rm ub*}_{\rm wst}(\mathbf{K},N,\mathbf{V})}\triangleq \min_{\mathbf{q}} \quad &  \wqb{R_{\rm wst}^{\rm ub}(\mathbf{K},N,\mathbf{V},\mathbf{q})}\nonumber \\
\mathrm{s.t.} \quad  &\eqref{eqn:q_range},\eqref{eqn:memory_constraint}.\nonumber
\end{align}
\end{Prob}
Let $\mathbf{q_{\rm wst}^\dag}$ denote an optimal solution of Problem~\ref{P:Worst Case max}.

Problem~\ref{P:Worst Case max} is a nonconvex problem with a differentiable objective function and linear constraint functions.
\wqb{A stationary point} of Problem~\ref{P:Worst Case max} can be obtained efficiently (e.g., using \wqb{gradient projection methods}~\cite{bertsekas1999nonlinear}).
As the numbers of variables and constraints \wqb{are both $O(NK)$}, \wqb{which are much smaller than those of Problem~\ref{P:Worst Case4}, the computational complexity for obtaining a \wqb{stationary point} of Problem~\ref{P:Worst Case max} is much lower than the computational complexity of Algorithm~\ref{alg:alg2}}. 
Note that we can run a \wqb{gradient projection} algorithm multiple times, each with a random initial point, and select the stationary point which achieves the minimum objective value for Problem~\ref{P:Worst Case max}.
Extensive numerical results show that this selected stationary point of Problem~\ref{P:Worst Case max} is usually an optimal solution of Problem~\ref{P:Worst Case max}.

Next, based on the upper and lower bounds in \eqref{ineqn:max_appro_ub} and \eqref{ineqn:max_appro_lb}, we characterize the \wqb{worst-case load increment} caused by the approximation, defined as \wqm{$L_{\rm wst}(\mathbf{K},N,\mathbf{V})\triangleq R_{\rm wst}(\mathbf{K},N,\mathbf{V},\mathbf{q}^\dag_{\rm wst})-R^*_{\rm wst}(\mathbf{K},N,\mathbf{V})$},
where $\mathbf{q_{\rm wst}^\dag}$ is an optimal solution of Problem~\ref{P:Worst Case max}.
\begin{Thm}[\wqb{Worst-case Load Increment}]\label{T:thm1}
	For all $c\ge 1$,
	\begin{align}
		L_{\rm wst}(\mathbf{K},N,\mathbf{V})&\le \frac{1}{c}\left(\sum_{i=1}^K \binom{K}{i}\ln i+K\ln N\right),\label{ineqn:T1_1}\\
		L_{\rm wst}(\mathbf{K},N,\mathbf{V})&=\frac{1}{c}O(2^K\ln K+K\ln N),\ {\rm{as}}\ K,N \to \infty.\label{ineqn:T1_2}
	\end{align}
\end{Thm}
\begin{proof}
Please refer to Appendix B.	
\end{proof}

From Theorem~\ref{T:thm1}, we know that the \wqb{worst-case load increment} can be made arbitrarily small by choosing a sufficiently large $c$, for \wqb{any} fixed $N$ and $K$.
However, in numerical experiments, the value of $c$ has to be kept modest in order not to exceed the maximum allowed value.
In addition, for all $c\ge 1$, the upper bound on the \wqb{worst-case load increment} \wqb{increases} with $N$ and $K$.
Whereas, it will be seen in Section~\ref{Sec:Numerical} that the \wqb{worst-case load of the \wqb{low-complexity} approximate solution} is still very promising for large $N$ or $K$.

\section{Average Load Minimization}\label{Sec:Average_Load_Minimization}
In this section, we first formulate a parameter-based coded caching design optimization problem over the considered class of schemes to minimize the average load by optimizing the caching parameter.
\wqb{Then}, we obtain a stationary point of an equivalent problem and a low-complexity approximate solution, \wqb{using methods similar to those in Section~\ref{Sec:Worst-case_Load_Minimization} for the worst-case load minimization.
We present the main results for completeness.}
To the best of our knowledge, this is the first \wqb{time that} optimization-based decentralized coded caching design \wqb{is obtained for reducing} the average load in the general scenario with arbitrary file sizes and cache sizes.

\subsection{Problem Formulation}\label{Subsec:Problem_Formulation(Ravg)}
We would like to optimize $\mathbf{q}$ subject to the constraints in \eqref{eqn:q_range} and \eqref{eqn:memory_constraint}
so as to minimize \wqb{$R_{\rm avg}(\mathbf{K},N,\mathbf{V},\mathbf{q})$}.
\begin{Prob}[Caching Parameter Optimization for \wqb{Reducing Average} Load]\label{P:Average Case1}
\begin{align}
	R^*_{\rm avg}(\mathbf{K},N,\mathbf{V})\triangleq \min_{\mathbf{q}} \quad &  R_{\rm avg}(\mathbf{K},N,\mathbf{V},\mathbf{q})\nonumber \\
	\mathrm{s.t.} \quad  &\eqref{eqn:q_range}, \eqref{eqn:memory_constraint}, \nonumber
\end{align}
\end{Prob}
\wqb{where $R_{\rm avg}(\mathbf{K},N,\mathbf{V},\mathbf{q})$ is given in \eqref{eqn:average_load_1}}. 

Problem~\ref{P:Average Case1} is a challenging problem \wqb{due to the nonconvexity and nondifferentiability of the objective function}.
\subsection{Solutions}\label{Subsec:Solutions(Ravg)}
\subsubsection{Stationary Point}\label{Subsubsec:Stationary_Point(Ravg)}
First, we obtain an equivalent \wqb{nonconvex} problem of Problem~\ref{P:Average Case1} \wqb{using a method similar to the one in} Section~\ref{Subsubsec:Stationary Point(Rwst)}.
%
\begin{Prob}[Equivalent Complementary GP of Problem~\ref{P:Average Case1}]\label{P:Average Case2}
 \begin{align}
 &\min_{\mathbf{q,x,w}\succ \mathbf{0}}  \quad \sum_{\mathbf{n}\in \mathcal{N}^K}\left(\prod_{j=1}^K p_{n_j}\right)\sum_{j=1}^K\sum_{\mathcal{S}\subseteq\mathcal{K}:|\mathcal{S}|=s } w_{\mathbf{n},\mathcal{S}}\nonumber\\
   &\quad \mathrm{s.t.} \quad\quad\  \eqref{eqn:wstcgp_1}, \eqref{eqn:wstcgp_2}, \eqref{eqn:rational function}, \eqref{eqn:wstcgp_4}.
   \end{align}
\end{Prob}
\begin{Lem}[Equivalence between Problem~\ref{P:Average Case1} and Problem~\ref{P:Average Case2}]\label{L:lem2} The optimal values of Problem~\ref{P:Average Case1} and Problem~\ref{P:Average Case2} are the same.
\end{Lem}
\begin{proof}
\wqb{The proof is similar to that of Lemma~\ref{L:lem1} and is omitted due to page limitation.}	
\end{proof}

\wqb{It is clear that Problem~\ref{P:Average Case2}} is a Complementary GP.
\wqb{We can obtain} a stationary point by solving a sequence of approximate GPs~\cite{chiang2007power}.
Specifically, at iteration $i$, $\left(\mathbf{q}^{(i)},\mathbf{x}^{(i)},\mathbf{w}^{(i)}\right)$ is updated by solving the following approximate GP of Problem~\ref{P:Average Case2}, which is parameterized by $\left(\mathbf{q}^{(i-1)},\mathbf{x}^{(i-1)},\mathbf{w}^{(i-1)}\right)$ obtained at iteration $i-1$.

\begin{Prob}[Approximate GP \wqb{of Problem~\ref{P:Average Case2}} at Iteration $i$]\label{P:Average Case3}
 \begin{align}
 \left(\mathbf{q}^{(i)},\mathbf{x}^{(i)},\mathbf{w}^{(i)}\right)
 \triangleq\ &\mathop{\arg\min}_{\mathbf{q,x,w}\succ \mathbf{0}}\ \sum_{\mathbf{n}\in \mathcal{N}^K}\left(\prod_{j=1}^K p_{n_j}\right)\sum_{j=1}^K\sum_{\mathcal{S}\subseteq\mathcal{K}:|\mathcal{S}|=s } w_{\mathbf{n},\mathcal{S}}\nonumber\\
   &\quad\ \ \ \mathrm{s.t.}\quad \eqref{eqn:wstcgp_1},\eqref{eqn:wstcgp_2},\eqref{eqn:wstcgp_4},\eqref{eqn:wstcgp_8}. \nonumber
 \end{align}
\end{Prob}

\wqb{The details are summarized in} Algorithm~\ref{alg:alg3}.
By \cite{marks1978general}, we know that  $(\mathbf{q}^{(i)},\mathbf{x}^{(i)},\mathbf{w}^{(i)})$ is provably convergent to a stationary point of Problem~\ref{P:Average Case2}, as $i\to \infty$. 
\begin{algorithm}[t]
    \caption{Algorithm for Obtaining a Stationary Point of Problem~\ref{P:Average Case2}}
        \begin{algorithmic}[1]
           \STATE \textbf{Initialization}: choose any feasible point $\left(\mathbf{q}^{(0)},\mathbf{x}^{(0)},\mathbf{w}^{(0)}\right)$ and set $i=1$\\
           \STATE \textbf{repeat}
           \STATE \quad Compute $\left(\mathbf{q}^{(i)},\mathbf{x}^{(i)},\mathbf{w}^{(i)}\right)$ by solving Problem \ref{P:Average Case3}  using 
           interior point methods
           \STATE\quad Set $i=i+1$
           \STATE \textbf{until} some convergence criterion is met
    \end{algorithmic}
    \label{alg:worst}\label{alg:alg3}
\end{algorithm}

\subsubsection{Low-complexity \wqb{Approximate} Solution}\label{Subsubsec:Low-complexity_Solution(Ravg)}
\wqb{As the numbers of variables and constraints of Problem~\ref{P:Average Case3} are $O\left( N^K \right)$ and $O\left( N^K K \right)$, respectively, and Problem~\ref{P:Average Case3} is solved at each iteration, the computational complexity of Algorithm~\ref{alg:alg3} is high for large $N$ or $K$.}
In this part, \wqb{we obtain a low-complexity approximate solution} of Problem~\ref{P:Average Case1}, which is applicable especially for large $N$ or $K$.

\wqb{First}, we approximate Problem~\ref{P:Average Case1} \wqb{which has} a nondifferentiable objective function $R_{\rm avg}(\mathbf{K},N,\mathbf{V},\mathbf{q})$ by a problem \wqb{whose objective function is a differentiable upper bound of $R_{\rm avg}(\mathbf{K},N,\mathbf{V},\mathbf{q})$)}, given by 
\begin{align}
	R_{\rm avg}^{\rm ub}(\mathbf{K},N,\mathbf{V},\mathbf{q})
	\triangleq &\sum_{\mathbf{n}\in \mathcal{N}^K}\left(\prod_{j=1}^K p_{n_j}\right)\sum_{s=1}^K\sum_{\mathcal{S}\subseteq\mathcal{K}:|\mathcal{S}|=s}\frac{1}{c}\ln\Bigg(\sum_{j\in\mathcal{S}}\exp\Bigg(c\Bigg(\prod_{a\in \mathcal{S}\backslash\{j\}}\prod_{\substack{t:a\in\mathcal{I}_t(\mathbf{K})}}q_{t,n_j}\Bigg)\nonumber\\
&\times\Bigg(\prod_{\substack{b\in\mathcal{K}\backslash  (\mathcal{S}\backslash\{j\})}}\prod_{t:b\in\mathcal{I}_t(\mathbf{K})}\left(1-q_{t,n_j}\right)\Bigg)V_{n_j}\Bigg)\Bigg).\label{ineq:relaxmodel_avg_ub}	
\end{align}

\begin{Prob}[Approximate Problem of Problem~\ref{P:Average Case1}]\label{P:Average Case max}
\begin{align}
 R^{\rm ub*}_{\rm avg}(\mathbf{K},N,\mathbf{M})\triangleq \min_{\mathbf{q}} \quad &  R_{\rm avg}^{\rm ub}(\mathbf{K},N,\mathbf{V},\mathbf{q})\nonumber \\
\mathrm{s.t.} \quad  &\eqref{eqn:q_range},\eqref{eqn:memory_constraint}.\nonumber
\end{align}
\end{Prob}
Let $\mathbf{q_{\rm avg}^\dag}$ denote an optimal solution of Problem~\ref{P:Average Case max}.

\wqb{The numbers of variables and constraints of Problem~\ref{P:Average Case max} are both $O\left( NK \right)$, which is much smaller than those of Problem~\ref{P:Average Case3}.
Thus, the computational complexity for obtaining a \wqb{stationary point} of Problem~\ref{P:Average Case max} (e.g., using \wqb{gradient projection methods}~\cite{bertsekas1999nonlinear}) is much lower than that of Algorithm~\ref{alg:alg3}.}

Next, we characterize the average load \wqb{increment} caused by the approximation, defined as $L_{\rm avg}(\mathbf{K},N,\mathbf{V})\triangleq R_{\rm avg}(\mathbf{K},N,\mathbf{V},\mathbf{q_{\rm avg}^\dag})-R^*_{\rm avg}(\mathbf{K},N,\mathbf{V})$,
where $\mathbf{q_{\rm avg}^\dag}$ is an optimal solution of Problem~\ref{P:Average Case max}.
\begin{Thm}[Average Load \wqb{Increment}]\label{T:thm2}
	For all $c\ge 1$,
	\begin{align}
		L_{\rm avg}(K,N,\mathbf{V},\mathbf{M})&\le \frac{1}{c}\sum_{i=1}^K \binom{K}{i}\ln i,\label{ineqn:T2_1}\\
		L_{\rm avg}(K,N,\mathbf{V},\mathbf{M})&=\frac{1}{c}O(2^K\ln K),\ {\rm{as}}\ K,N \to \infty.\label{ineqn:T2_2}
	\end{align}
\end{Thm}
\begin{proof}
Please refer to Appendix C.	
\end{proof}

From Theorem~\ref{T:thm2}, we know that the average load \wqb{increment} can be made arbitrarily small by choosing a sufficiently large $c$, for \wqb{any} fixed $N$ and $K$.
It will be seen in Section~\ref{Sec:Numerical} that the \wqb{average load of the low-complexity approximate solution} is still very promising for large $N$ or $K$.

\section{Converse Bound}\label{Sec:Converse}
In this section, we present information-theoretic converse bounds on the minimum worst-case load \wqb{and average load (under an arbitrary file popularity) in the general scenario, respectively.}
Both converse bounds belong to the second class of converse bounds stated in Section~\ref{Sec:Introduction} (applicable to \wqb{both} uncoded placement and coded placement).
\wqb{To our knowledge, this is the first work providing information-theoretic converse bounds on the minimum worst-case load and average load in the general scenario.}

\wqb{Given $\mathcal{S}\subseteq\mathcal{L},|\mathcal{S}|=s$, let $M_{[\mathcal{S},1]}\le M_{[\mathcal{S},2]}\le \dots \le M_{[\mathcal{S},s-1]}\le M_{[\mathcal{S},s]}$ be $M_l,l\in\mathcal{S}$ arranged in increasing order, so that $M_{[\mathcal{S},i]}$ is the $i$-th smallest.}
Extending the proof \wqb{for} the converse bound on the minimum worst-case load \wqb{in the scenario with the same file size and cache size} \wqm{in}~\cite{Chien_Yi_Wang2018improved}, we \wqb{obtain a converse bound on \wqb{that} in the scenario with arbitrary file sizes and cache sizes.}
\begin{Lem}[Converse Bound on Worst-case Load]\label{L:converse_bound_wst}
	\ \ \ The\ \ \ minimum\ \ \ worst-case\ \ \ load\ \ \ $R_{\rm wst}^*(\mathcal{L}_a,N,\mathbf{V},\mathbf{M})$ for the shared, error-free link caching network with \wqb{active users in $\mathcal{L}_a$}, library size $N\in\mathbb{N}$, file sizes $\mathbf{V}$ and cache sizes $\mathbf{M}$ satisfies
	\begin{align}
		 R_{\rm wst}^*(\mathcal{L}_a,N,\mathbf{V},\mathbf{M})\ge\max_{m\in\{1,\dots,\min\{N,L_a\}\}}\left\{\frac{m}{N}\sum_{i=1}^N V_{i}-\min\left\{\sum_{l=1}^{m} \frac{\sum_{i=1}^l M_{[\mathcal{L}_a,i]}}{N-l+1},\frac{m}{N}\sum_{i=1}^{m} M_{[\mathcal{L}_a,i]}\right\}\right\}.\label{ineqn:converse_bound_wst}
 	\end{align}
\end{Lem}
\begin{proof}
Please refer to Appendix D.	
\end{proof}

\wqb{It is clear} that the converse bound on the minimum worst-case load for arbitrary file sizes and  cache sizes reduces to the one in~\cite{Chien_Yi_Wang2018improved} when $V_n,n\in\mathcal{N}$ are the same and $M_l,l\in\mathcal{L}_a$ are the same.

\wqb{In addition, extending the proof for the converse bound on the minimum average load in the scenario with the same file size and cache size \wqm{in}~\cite{jin}, which rests on \wqb{the genie-aided approach in~\cite{ji2015order} and the proof \wqb{of Theorem 3} in~\cite{Chien_Yi_Wang2018improved}}, we have the following result.}
\begin{Lem}[Converse Bound on Average Load]\label{L:converse_bound_avg}
The minimum average load \wqb{$R_{\rm avg}^*(\mathcal{L}_a,N,\mathbf{V},\mathbf{M})$} for the shared, error-free link caching network with \wqb{active users in $\mathcal{L}_a$}, library size $N\in\mathbb{N}$, file sizes $\mathbf{V}$ and cache sizes $\mathbf{M}$ under \wqb{arbitrary} file popularity distribution $\mathbf{p}$ satisfies
	\begin{align}
	 R_{\rm avg}^{*}(\mathcal{L}_a,N,\mathbf{V},\mathbf{M})\ge \max_{N^{'}\in\{1,\dots,N\}}\sum_{i=1}^{L_a} \left(N^{'}p_{N^{'}}\right)^i \left(1-N^{'}p_{N^{'}}\right)^{L_a-i}
		 \sum_{\substack{\mathcal{S}\subseteq\mathcal{L}_a: |\mathcal{S}|=i}} R_{\rm avg,unif}^{\rm lb}\left(\mathcal{S},N^{'},\mathbf{V},\mathbf{M}\right),\label{ineqn:converse_bound_avg}
	\end{align}
\end{Lem}
where
\begin{small}
\begin{align}
	&R_{\rm avg,unif}^{\rm lb}(\mathcal{S},N^{'},\mathbf{V},\mathbf{M})\triangleq\nonumber\\
		 & \max_{m\in\{1,\dots,\min\{N^{'},L_a\}\}}\Bigg\{\sum_{j=1}^{m}\frac{\binom{N^{'}-1}{j-1}j!\stirling{m}{j}}{{(N^{'})}^{m}}\sum_{i=1}^{N^{'}} V_{i}
		 -\min\Bigg\{\sum_{l=1}^{m}\frac{\sum_{i=1}^l M_{[\mathcal{S},i]}}{N^{'}},	\Bigg(1-\left(1-\frac{1}{N^{'}}\right )^{m}\Bigg)\sum_{i=1}^{m} M_{[\mathcal{S},i]}\Bigg\}	 \Bigg\}.\label{ineqn:converse_bound_avg_unif}
\end{align}
\end{small}	
\wqb{Here $\stirling{m}{j}$ denotes} the Stirling number of the second kind.
\begin{proof}
Please refer to Appendix E.	
\end{proof}

\wqb{In Lemma~\ref{L:converse_bound_avg}, $R_{\rm avg}^{\rm lb}(\mathcal{L}_a,N,\mathbf{V},\mathbf{M})$ and $R_{\rm avg,unif}^{\rm lb}(\mathcal{L}_a,N,\mathbf{V},\mathbf{M})$ represent the converse bounds on the average load under an arbitrary file popularity and the uniform file popularity, respectively.}
\wqb{It is clear that} the converse bound on the minimum average load for arbitrary file sizes and  cache sizes under an arbitrary file popularity reduces to the one in~\cite{jin} when $V_n,n\in\mathcal{N}$ are the same and $M_l,l\in\mathcal{L}_a$ are the same, and reduces to the one in~\cite{Chien_Yi_Wang2018improved} when $V_n,n\in\mathcal{N}$ are the same, $M_l,l\in\mathcal{L}_a$ are the same and the file popularity is uniform.
\section{Numerical Results}\label{Sec:Numerical}
\wqb{In this section, we compare the proposed solutions in Section~\ref{Sec:Worst-case_Load_Minimization} and Section~\ref{Sec:Average_Load_Minimization} with the existing ones \cite{AliDec,wang2015coded, zhang2015coded,cheng2017optimal,NonuniformDemands,ji2015order,zhang2018coded,Sinong} and the derived converse bounds in Section~\ref{Sec:Converse}.
As the computational \wqb{complexities} for the proposed stationary points \wqb{are} high for large $N$ or $K$, their performances are evaluated only at $N\le 4$ or $K\le 4$.
When obtaining the low-complexity approximate solutions, we choose $c=1$.
In addition, we adopt the existing solutions in \cite{AliDec,wang2015coded,NonuniformDemands,ji2015order,zhang2018coded,Sinong} under the file size $\max_{n\in\mathcal{N}}\{V_n\}$, and the existing solutions in \cite{AliDec,zhang2015coded,cheng2017optimal,NonuniformDemands,ji2015order,zhang2018coded,Sinong} under the cache size \wqb{$\min_{t\in\mathcal{T}}\{\overline{M}_t\}$}, as they are originally proposed for the scenarios with the same file size ($V_n,n\in\mathcal{N}$ are the same) \wqb{and} the same cache size \wqb{($\overline{M}_t,t\in\mathcal{T}$ are the same), respectively}.}

\wqb{In the simulation, for ease of illustration, we assume file sizes (in $10^3$ data units) and cache sizes (in $10^3$ data units) are arithmetic sequences, i.e., $V_n=V_1+(n-1)\Delta V,n\in\mathcal{N}$ and $\overline{M}_t=\overline{M}_1+(k-1)\Delta \overline{M},t\in\mathcal{T}$, where $\Delta V$ and $\Delta \overline{M}$ represent the common differences of the two sequences, respectively.
\wqb{In addition,} in the average case, we assume \wqb{that} the file popularity follows Zipf distribution, i.e., $p_n=\frac{n^{-\gamma}}{\sum_{n\in\mathcal{N}}n^{-\gamma}}$ for all $n\in\mathcal{N}$, where $\gamma$ is the Zipf parameter~\cite{NonuniformDemands,ji2015order,zhang2018coded,Sinong}.}
\wqb{We also consider different choices for parameter $\mathbf{K}$ for the proposed solutions and those in~\cite{ji2015order,Sinong} and \cite{cheng2017optimal}.}
\begin{figure}
\begin{center}
  \subfigure[$V_1=10,\ \Delta V=-1,\ T=4,\ \overline{M}_1=5.5$ and $\Delta \overline{M}=5.5$.]
  {\resizebox{7cm}{!}{\includegraphics{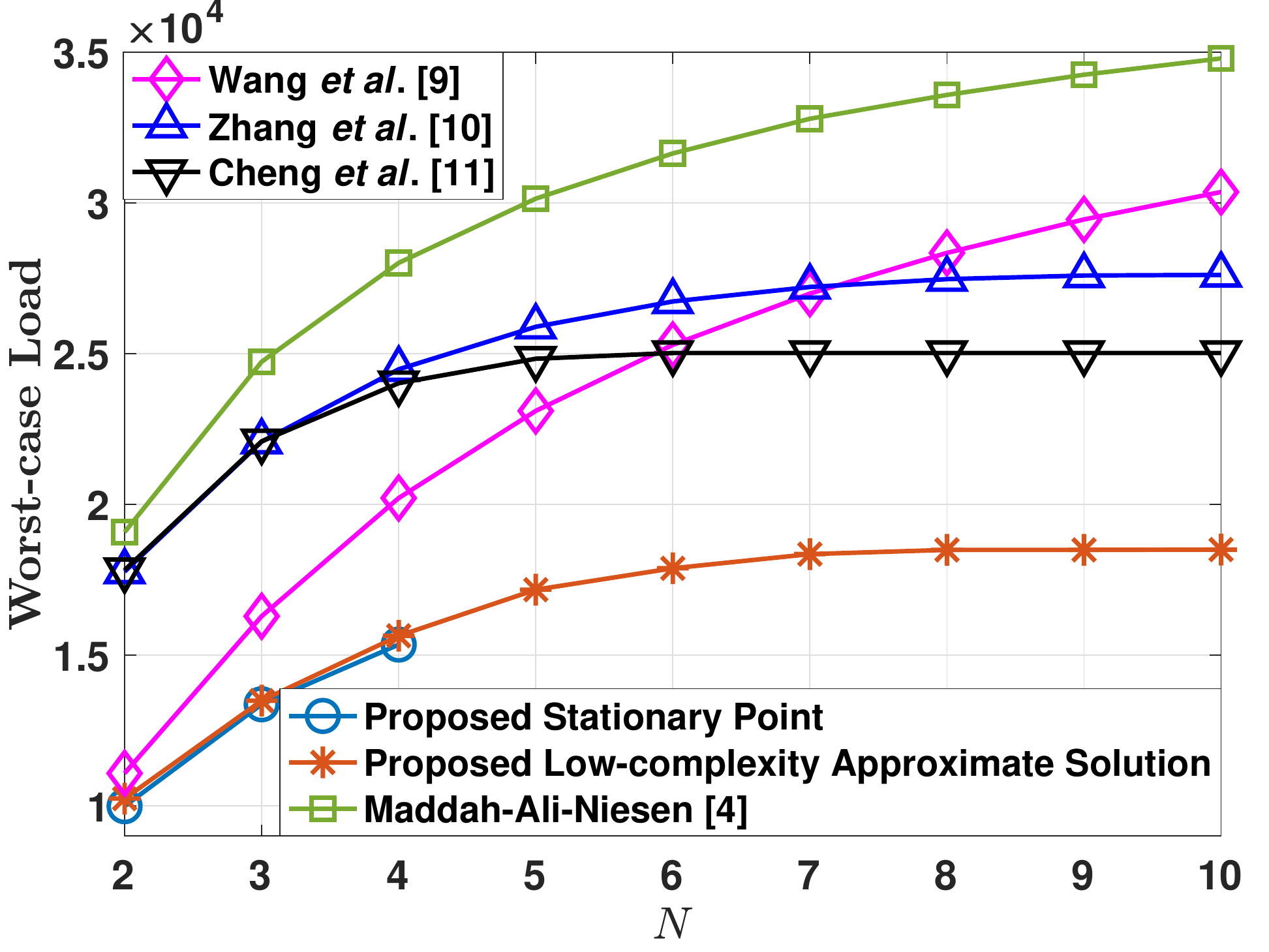}}}
  \quad
  \subfigure[$\ N=4,\ V_1=13,\ \Delta V=-4,\ \overline{M}_1=5$ and $\Delta \overline{M}=1$.]
  {\resizebox{7cm}{!}{\includegraphics{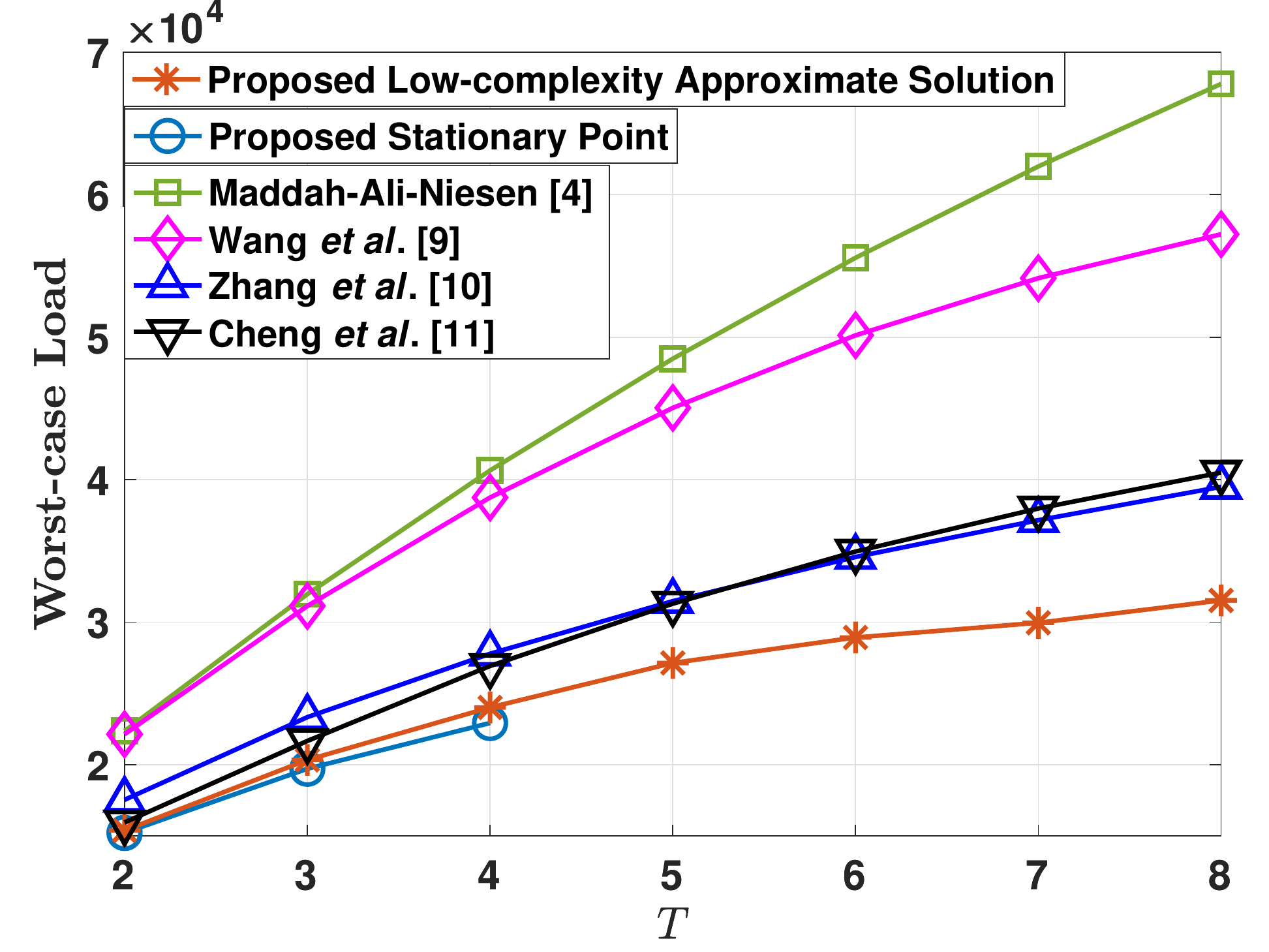}}}
  \end{center}
         \caption{\small{Worst-case load versus $N$ and $T$ when
         $\mathbf{L}_a=(L_{a,t})_{t\in\mathcal{T}}$ with $L_{a,t}=1,\ t\in\mathcal{T}$ and $\mathbf{K}=\mathbf{L}_a$.}
         \vspace{-8mm}
         }\label{fig:simulation1}
\end{figure}
\begin{figure}
\begin{center}
  \subfigure[$V_1=25.5-24.5\Delta V,\ \overline{M}_1=140$ and $\Delta \overline{M}=120$.]
  {\resizebox{7cm}{!}{\includegraphics{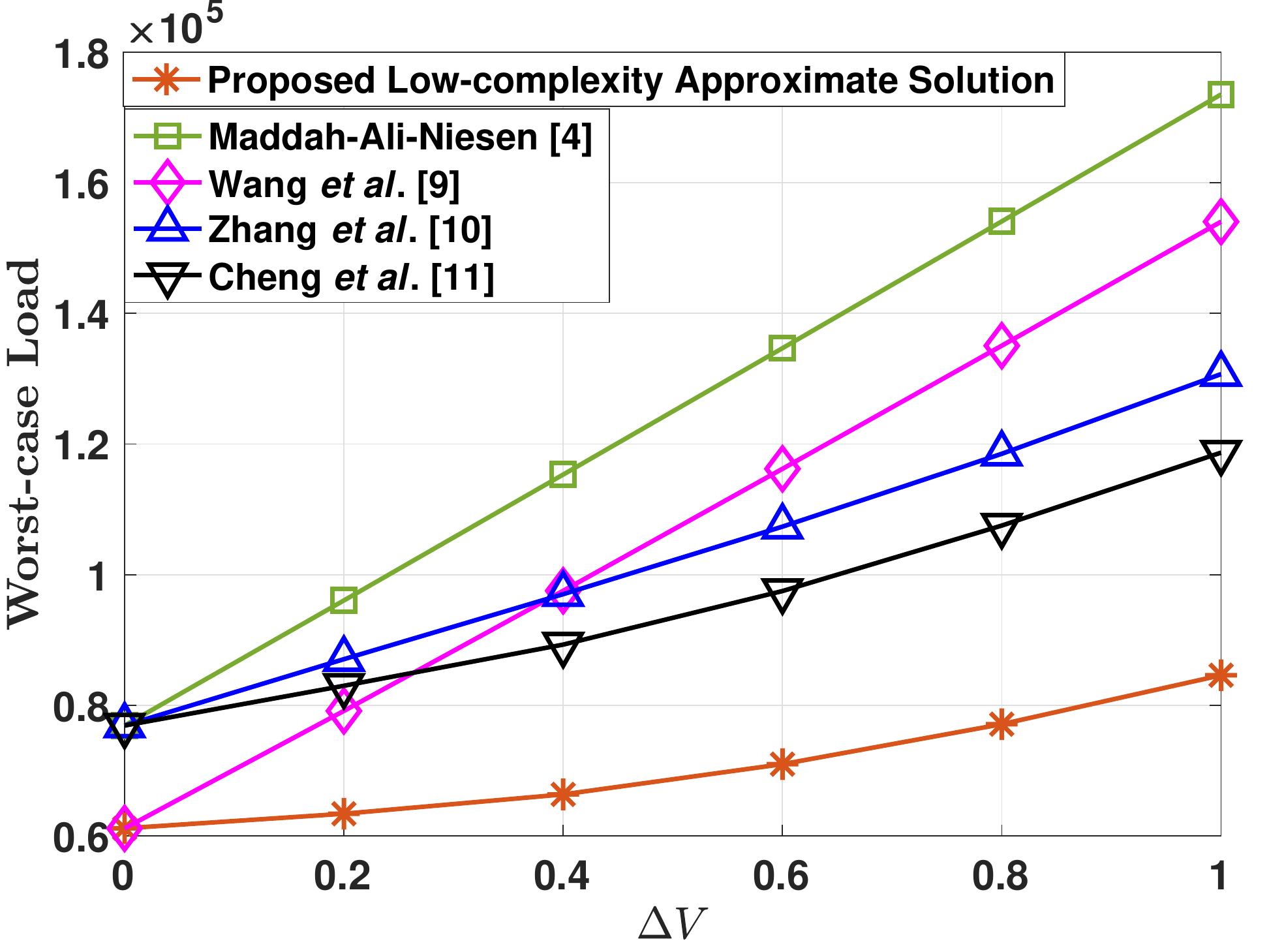}}}
  \quad
  \subfigure[$V_1=1,\ \Delta V=1$ and $\overline{M}_1=318.75-1.5\Delta \overline{M}$.]
  {\resizebox{7cm}{!}{\includegraphics{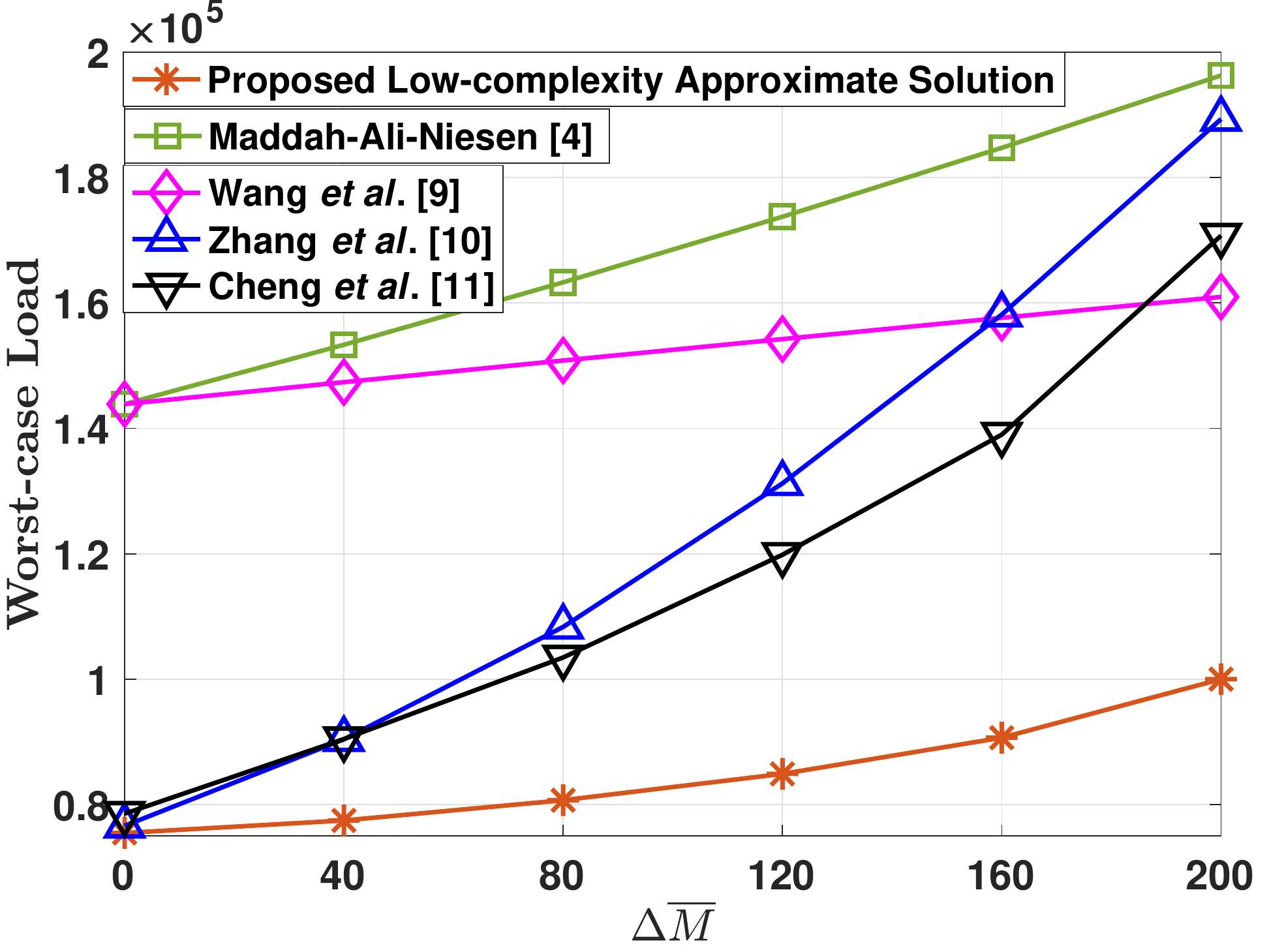}}}
  \end{center}
  		  \vspace{-4mm}
         \caption{\small{Worst-case load versus $\Delta  V$ and $\Delta  \overline{M}$ when 
         $N=50,\ T=4,\ \mathbf{L}_a=(L_{a,t})_{t\in\mathcal{T}}$ with $L_{a,t}=1,\ t\in\mathcal{T}$ and $\mathbf{K}=\mathbf{L}_a$.}
         \vspace{-8mm}
         }\label{fig:simulation2}
\end{figure}
\begin{figure}
\begin{center}
  \subfigure[$N=50,\ V_1=1,\ \Delta V=1,\ T=4,\ \overline{M}_1=80\overline{M}_0,\ \Delta \overline{M}=160\overline{M}_0,\ \mathbf{L}_a=(L_{a,t})_{t\in\mathcal{T}}$ with $L_{a,t}=1,\ t\in\mathcal{T}$ and $\mathbf{K}=\mathbf{L}_a$.]
  {\resizebox{7cm}{!}{\includegraphics{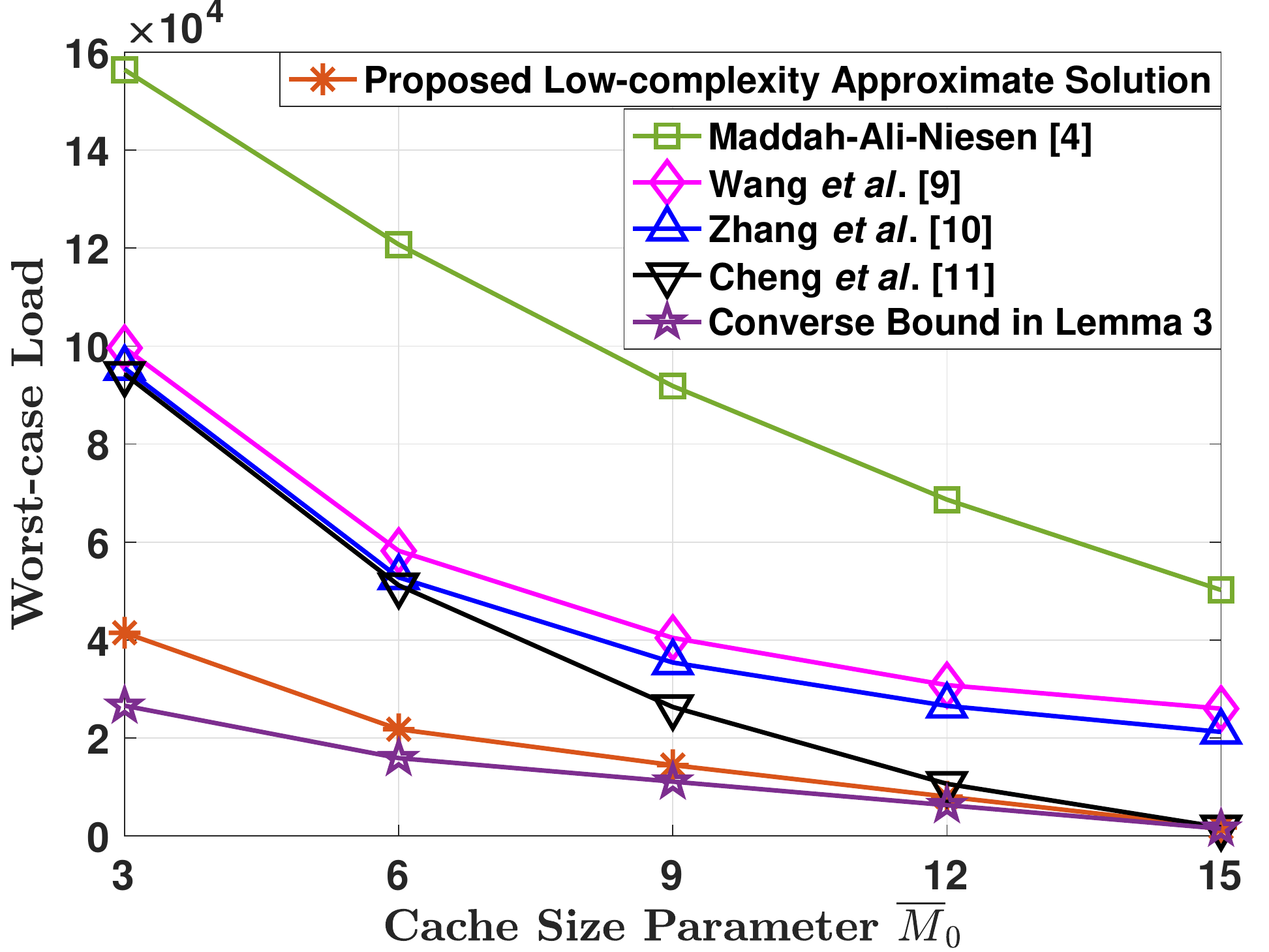}}}
  \quad
  \subfigure[$N=10,\ V_1=10,\ \Delta V=-1,\ T=2,\ \overline{M}_1=5\overline{M}_0,\ \overline{M}_2=20\overline{M}_0,\ L_1=4,\ L_2=2$ and $\mathbf{K}=(K_t)_{t\in\mathcal{T}}$ with $K_t=0.5L_t,\ t\in\mathcal{T}$.
  			\wqb{Note that we assume that each user is active with probability 0.5, and show the average (over random $\mathbf{L}_a$) of the worst-case load of each scheme and the average of the converse bound in Lemma \ref{L:converse_bound_wst}.}]
  {\resizebox{7cm}{!}{\includegraphics{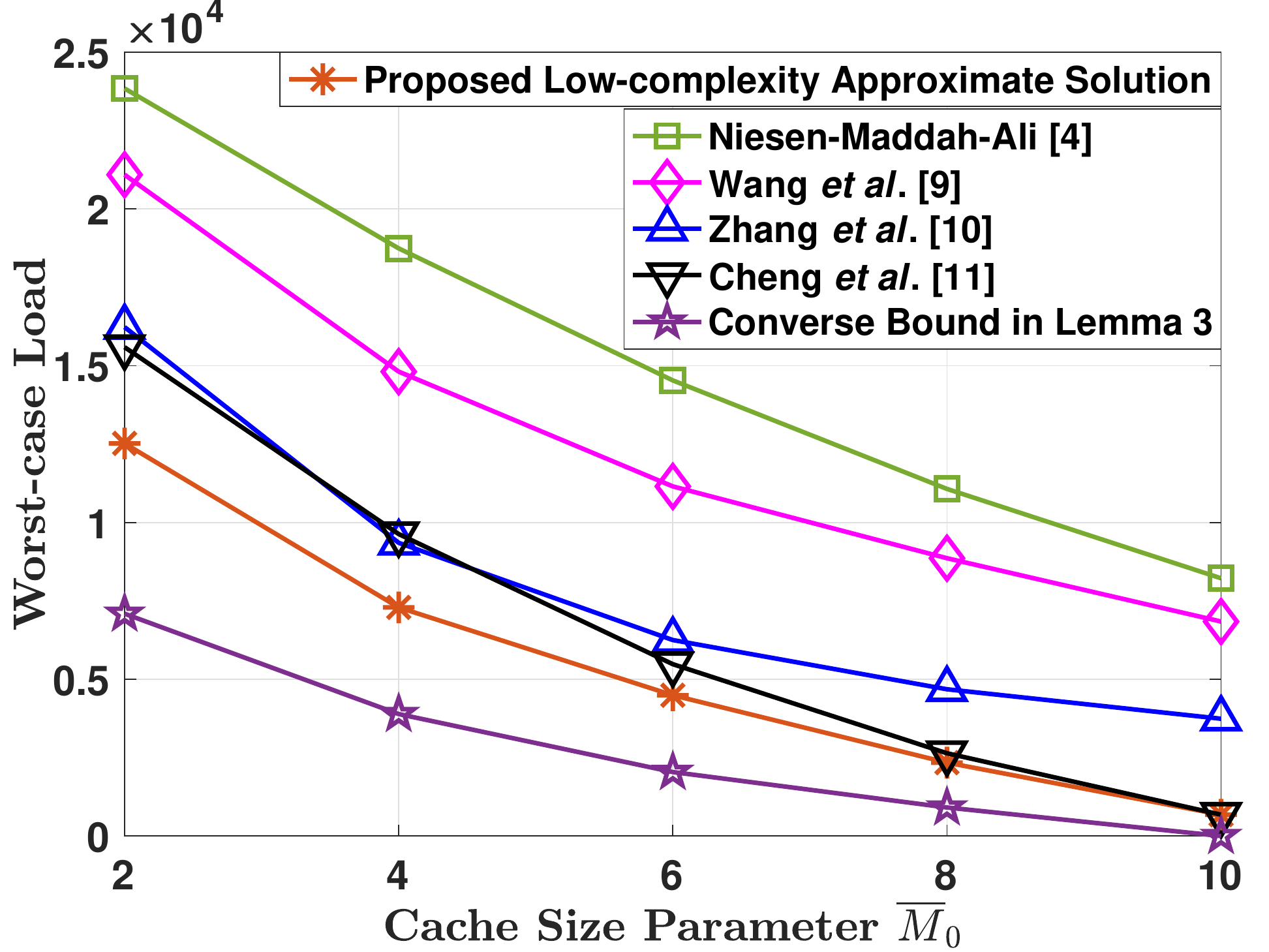}}}
  \end{center}
         \vspace{-4mm}
         \caption{\small{Worst-case load versus cache size parameter $\overline{M}_0$.}
         \vspace{-8mm}
         }\label{fig:simulation3}
\end{figure}
\subsection{Comparison of Worst-case Loads}\label{Subsec:Numerical_Comparison_wst}
In this part, we compare \wqb{the worst-case loads of} the proposed stationary point and low-complexity \wqb{approximate} solution \wqb{in Section~\ref{Sec:Worst-case_Load_Minimization}} with \wqb{those of} the existing solutions in~\cite{AliDec,wang2015coded,zhang2015coded,cheng2017optimal} and the converse bound in Lemma~\ref{L:converse_bound_wst}.

Fig.~\ref{fig:simulation1} (a) and Fig.~\ref{fig:simulation1} (b) illustrate the worst-case loads of the six schemes versus $N$ and \wqb{$T$ (which is the same as $K$ in Fig.~\ref{fig:simulation1})}, respectively.
From Fig.~\ref{fig:simulation1}, we can see that the performance gap between the \wqb{proposed} stationary point and low-complexity \wqb{approximate} solution is rather small, 
which shows a promising prospect of our low-complexity \wqb{approximate} solution.

Fig.~\ref{fig:simulation2} (a) and Fig.~\ref{fig:simulation2} (b) illustrate the worst-case loads of the proposed low-complexity \wqb{approximate} solution and the four baseline schemes versus $\Delta V$ and \wqb{$\Delta \overline{M}$}, respectively.
From Fig.~\ref{fig:simulation2}, we can see that the \wqb{worst-case} load of the decentralized coded caching scheme in~\cite{AliDec} increases rapidly with both $\Delta V$ and \wqb{$\Delta \overline{M}$}.
The \wqb{worst-case} load of the decentralized coded caching scheme in~\cite{wang2015coded} increases rapidly with $\Delta V$, but increases slowly with \wqb{$\Delta \overline{M}$}. 
In contrast, the \wqb{worst-case} loads of the decentralized coded caching schemes in~\cite{zhang2015coded} and \cite{cheng2017optimal} increase slowly with $\Delta V$, but increase rapidly with \wqb{$\Delta \overline{M}$}.
Note that the \wqb{worst-case} load of the proposed low-complexity \wqb{approximate} solution increases slowly with both $\Delta V$ and \wqb{$\Delta \overline{M}$}, \wqb{indicating that it well adapts to the changes of file sizes and cache sizes.}

Fig.~\ref{fig:simulation3} (a) and Fig.~\ref{fig:simulation3} (b) illustrate the worst-case loads of the proposed low-complexity \wqb{approximate} solution \wqb{and} the four baseline schemes, and the converse bound \wqb{in Lemma~\ref{L:converse_bound_wst}} versus \wqb{$\overline{M}_0$}, which determines \wqb{$\left(\overline{M}_t\right)_{t\in\mathcal{T}}$} as illustrated in the caption of Fig.~\ref{fig:simulation3}.
From Fig.~\ref{fig:simulation3}, we can see that the \wqb{worst-case} load of the low-complexity \wqb{approximate} solution is close to the converse bound \wqb{in Lemma~\ref{L:converse_bound_wst}}, implying that it is close to optimal.

From Fig.~\ref{fig:simulation1}, Fig.~\ref{fig:simulation2} and Fig.~\ref{fig:simulation3}, we can see that, the two proposed solutions outperform the four baseline schemes in~\cite{AliDec,wang2015coded,zhang2015coded,cheng2017optimal} at the system parameters considered in the simulation.
Their gains over \wqb{the one} in \cite{AliDec} are due to the adaptation to arbitrary file sizes and cache sizes;
their gains over \wqb{the one} in \cite{wang2015coded} follow by the consideration of arbitrary file sizes;
and their gains over \wqb{the schemes} in \cite{zhang2015coded} and \cite{cheng2017optimal} follow by the consideration of arbitrary cache sizes.

\begin{figure}
\begin{center}
  \subfigure[$V_1=20,\ \Delta V=-1,\ \gamma=1.2,\ T=4,\ \overline{M}_1=5.5$ and $\Delta \overline{M}=5.5$.]
  {\resizebox{7cm}{!}{\includegraphics{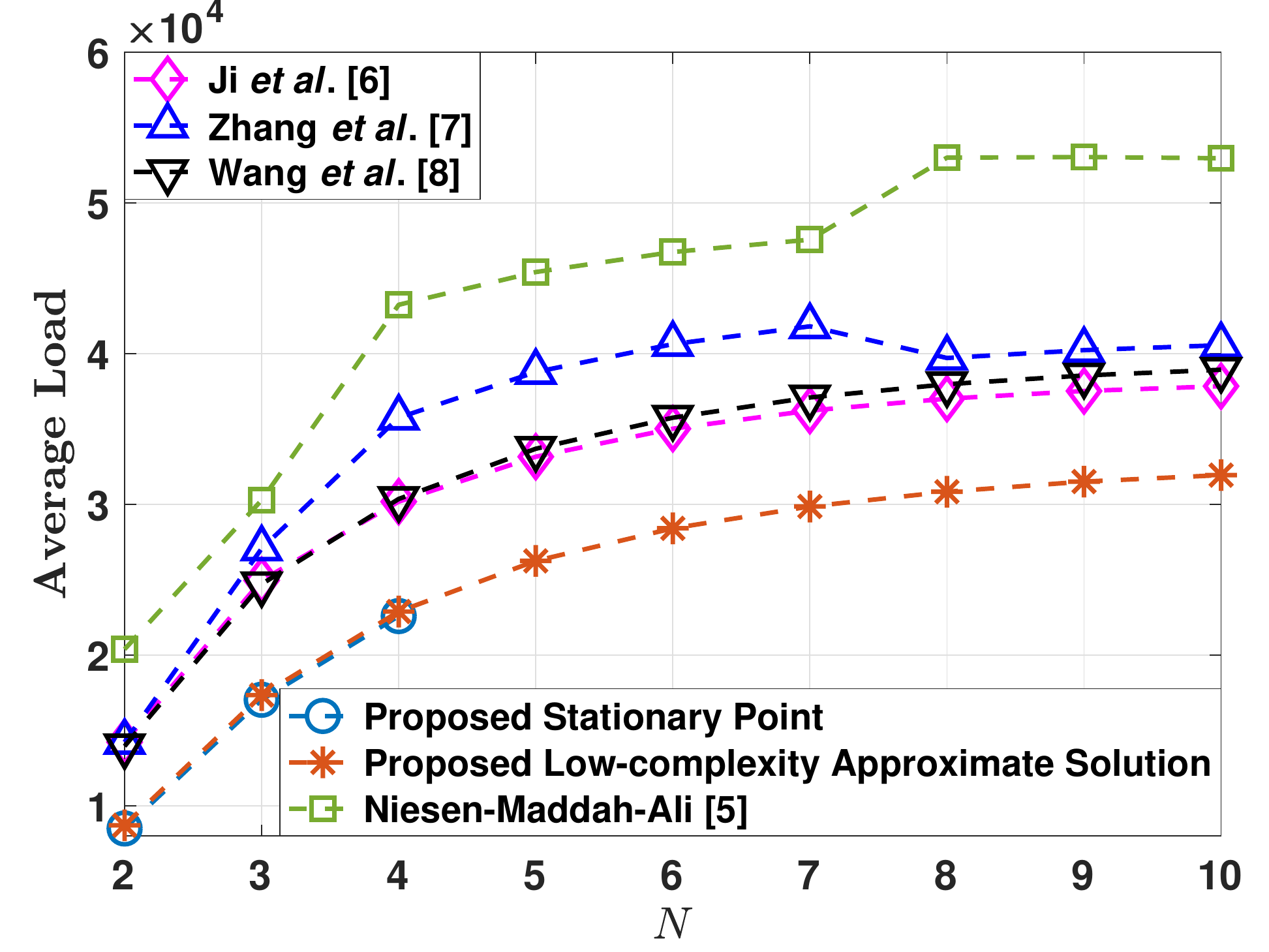}}}
  \quad
  \subfigure[$N=4,\ V_1=23,\ \Delta V=-4,\ \gamma=1.2,\ \overline{M}_1=5$ and $\Delta \overline{M}=1$.]
  {\resizebox{7cm}{!}{\includegraphics{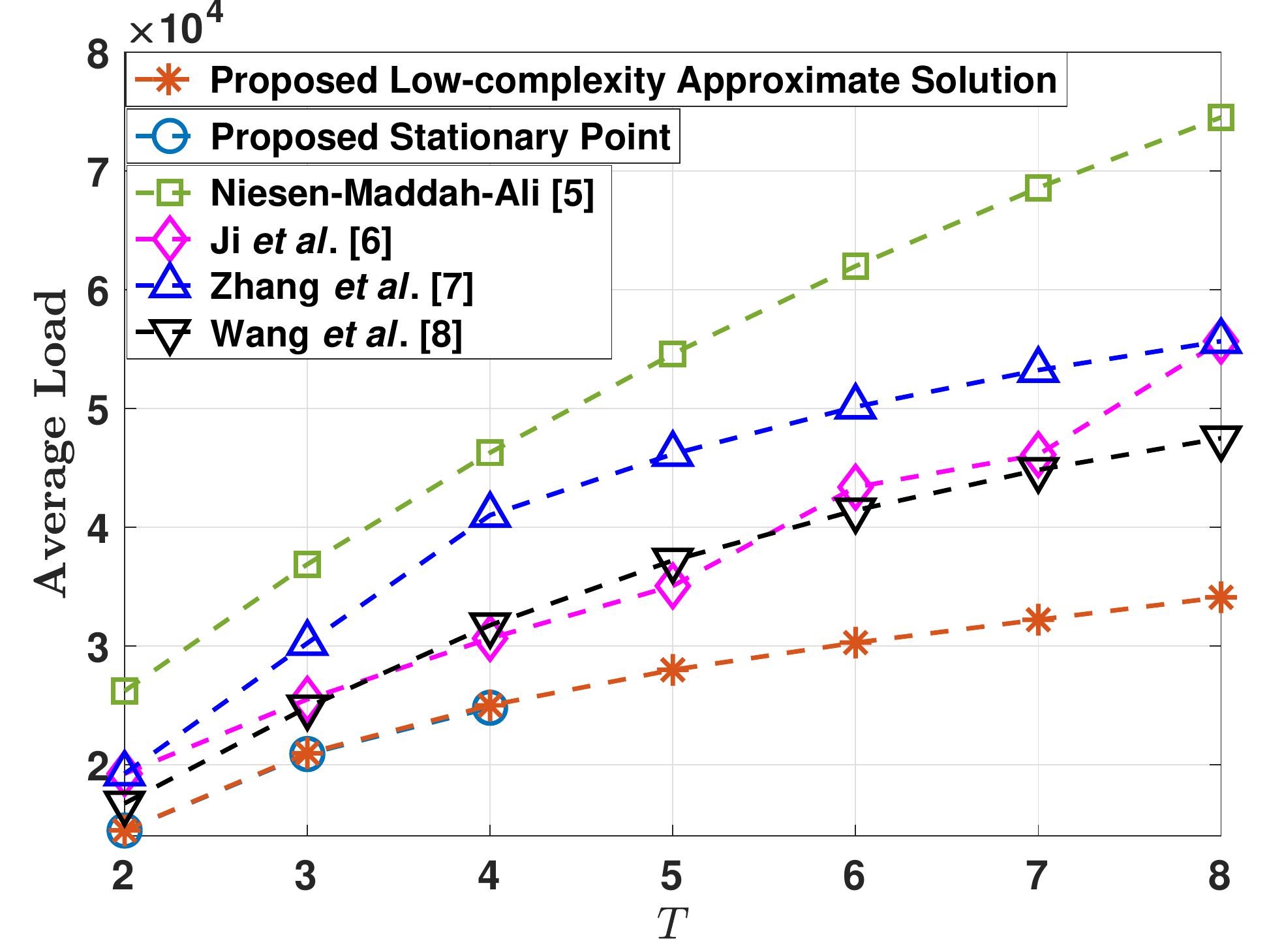}}}
  \end{center}
         \caption{\small{Average load versus $N$ and $T$ when
         $\mathbf{L}_a=(L_{a,t})_{t\in\mathcal{T}}$ with $L_{a,t}=1,\ t\in\mathcal{T}$ and $\mathbf{K}=\mathbf{L}_a$.}
         }\label{fig:simulation_avg1}
\end{figure}
\begin{figure}
\begin{center}
  \subfigure[$V_1=25.5-24.5\Delta V,\ \gamma=1.2,\ \overline{M}_1=140$ and $\Delta \overline{M}=120$.]
  {\resizebox{7cm}{!}{\includegraphics{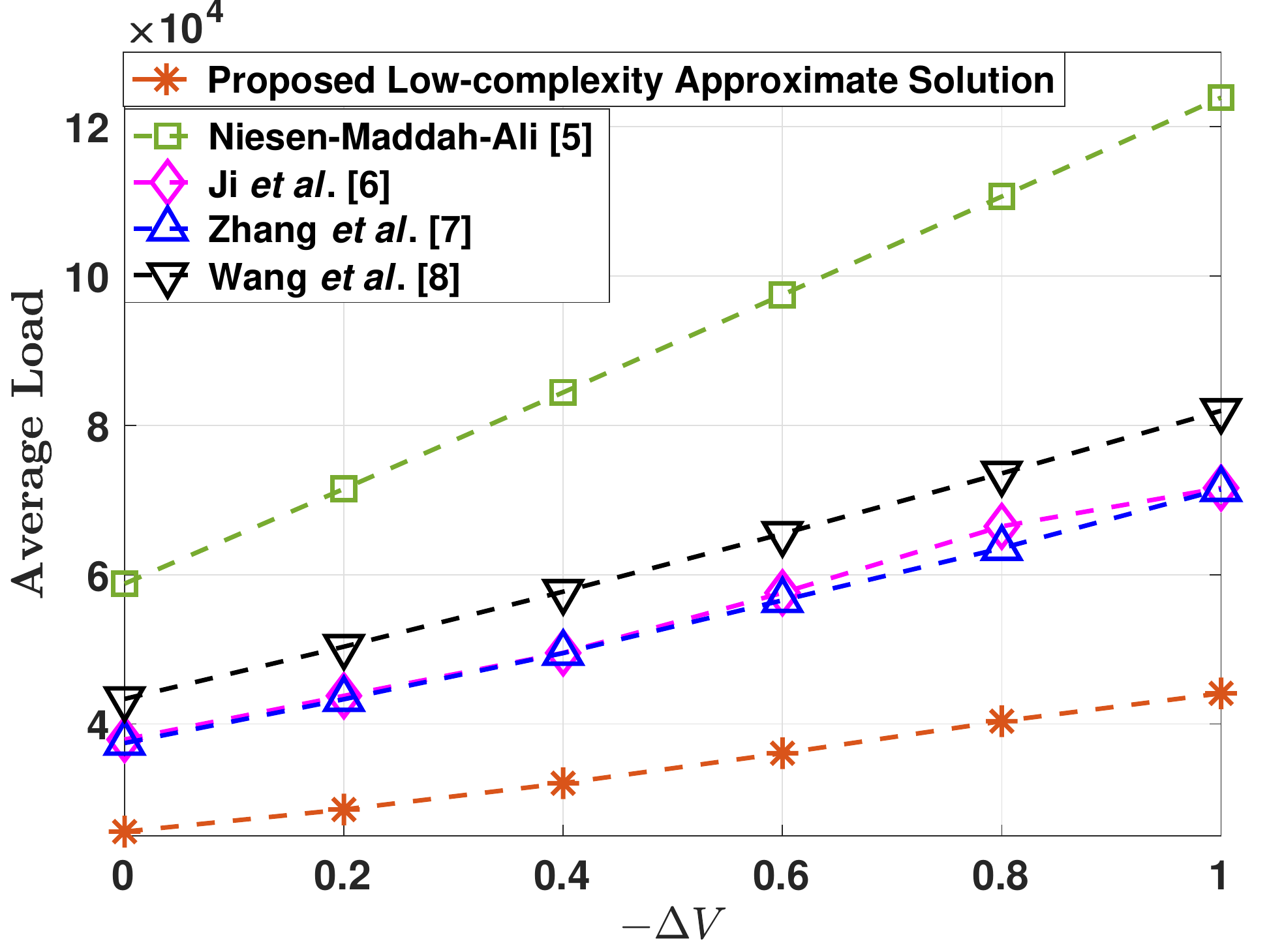}}}
  \quad
  \subfigure[$V_1=50,\ \Delta V=-1,\ \gamma=1.2$ and $\overline{M}_1=318.75-1.5\Delta \overline{M}$.]
  {\resizebox{7cm}{!}{\includegraphics{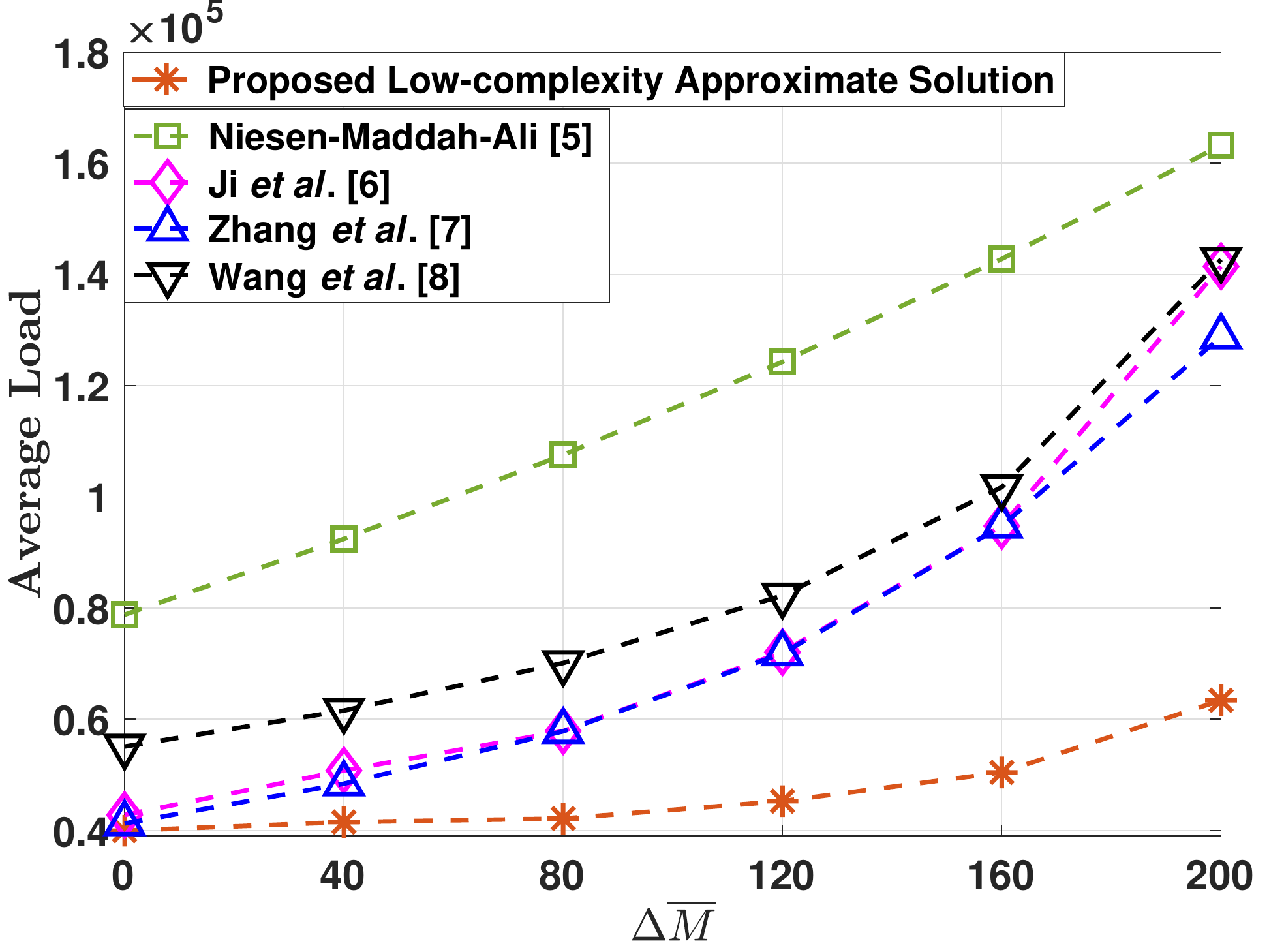}}}
  \end{center}
         \caption{\small{Average load versus $\Delta V$ and $\Delta  \overline{M}$ when
         $N=50,\ T=4,\ \mathbf{L}_a=(L_{a,t})_{t\in\mathcal{T}}$ with $L_{a,t}=1,\ t\in\mathcal{T}$ and $\mathbf{K}=\mathbf{L}_a$.}
         \vspace{-8mm}
         }\label{fig:simulation_avg2}
\end{figure}
\begin{figure}
\begin{center}
  {\resizebox{7cm}{!}{\includegraphics{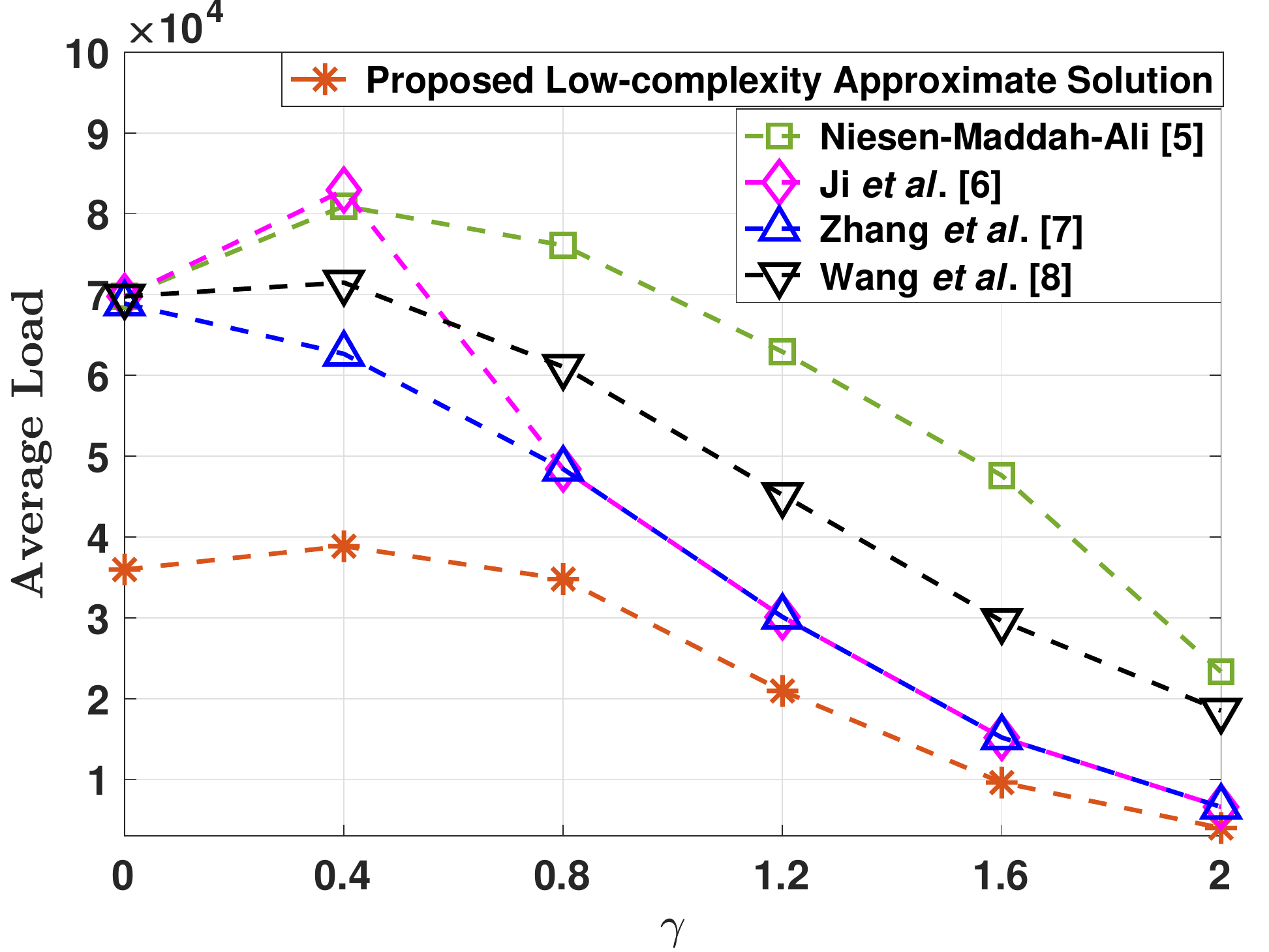}}}
         \caption{\small{Average load versus $\gamma$ when $N=50,\ V_1=50,\ \Delta V=-1,\ T=4,\ \overline{M}_1=450,\ \Delta \overline{M}=90,\ \mathbf{L}_a=(L_{a,t})_{t\in\mathcal{T}}$ with $L_{a,t}=1,\ t\in\mathcal{T}$ and $\mathbf{K}=\mathbf{L}_a$.}
         }\label{fig:simulation_avg3}
\end{center}
\end{figure}
\begin{figure}
\begin{center}
  \subfigure[$N=50,\ V_1=50,\ \Delta V=-1,\ \gamma=1,\ T=4,\ \overline{M}_1=50\overline{M}_0,\ \Delta \overline{M}=10\overline{M}_0,\ \mathbf{L}_a=(L_{a,t})_{t\in\mathcal{T}}$ with $L_{a,t}=1,\ t\in\mathcal{T}$ and $\mathbf{K}=\mathbf{L}_a$]
  {\resizebox{7cm}{!}{\includegraphics{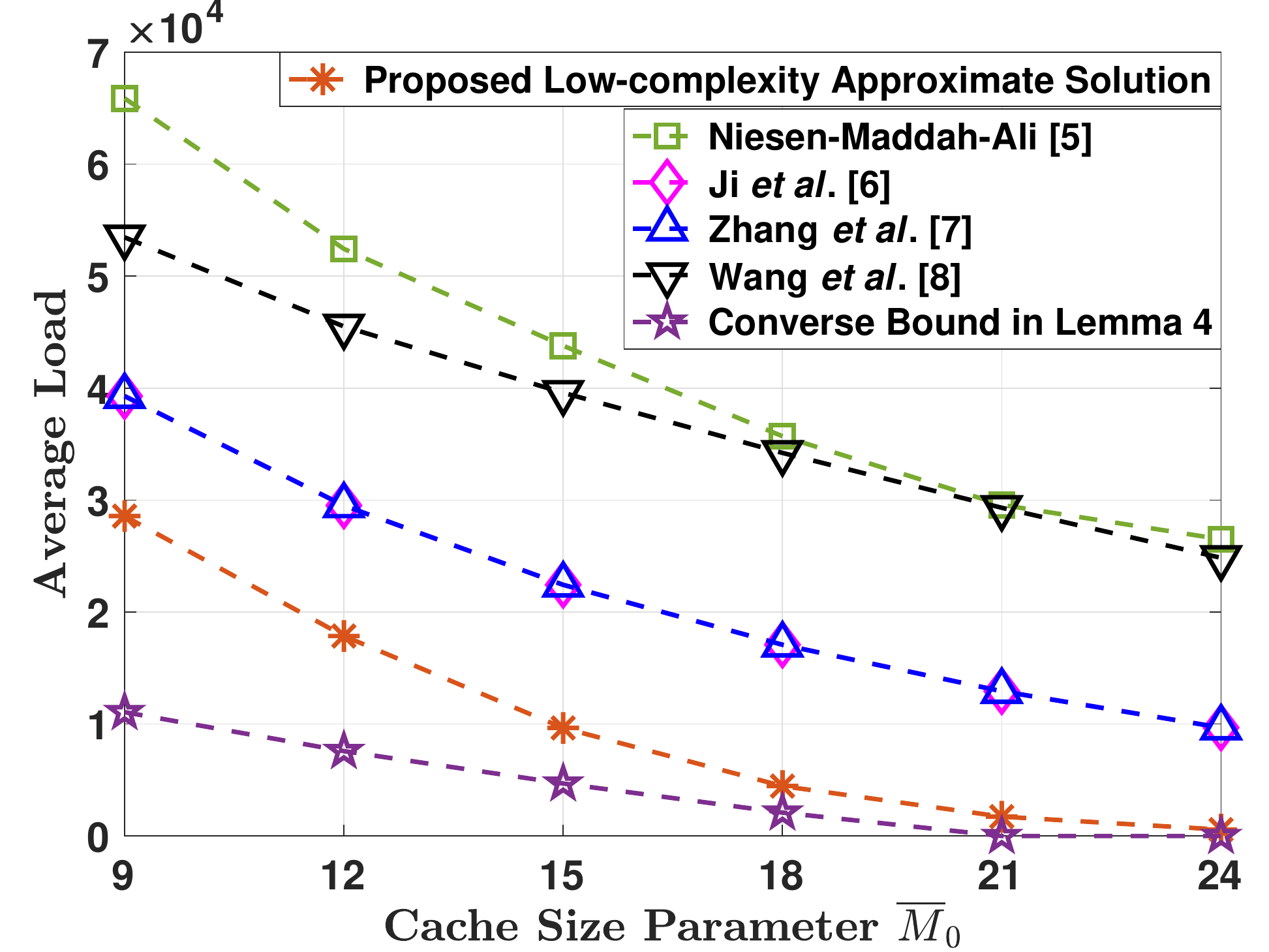}}}
  \quad
  \subfigure[$N=10,\ V_1=10,\ \Delta V=-1,\ \gamma=1.5,\ T=2,\ \overline{M}_1=5\overline{M}_0,\ \overline{M}_2=20\overline{M}_0,\ L_1=4,\ L_2=2$ and $\mathbf{K}=(K_t)_{t\in\mathcal{T}}$ with $K_t=0.5L_t, t\in\mathcal{T}$.
  			\wqb{We assume that each user is active with probability 0.5, and show the \wqb{expectations} (over random $\mathbf{L}_a$) of the average load of each scheme and the converse bound in Lemma \ref{L:converse_bound_avg}.}]
  {\resizebox{7cm}{!}{\includegraphics{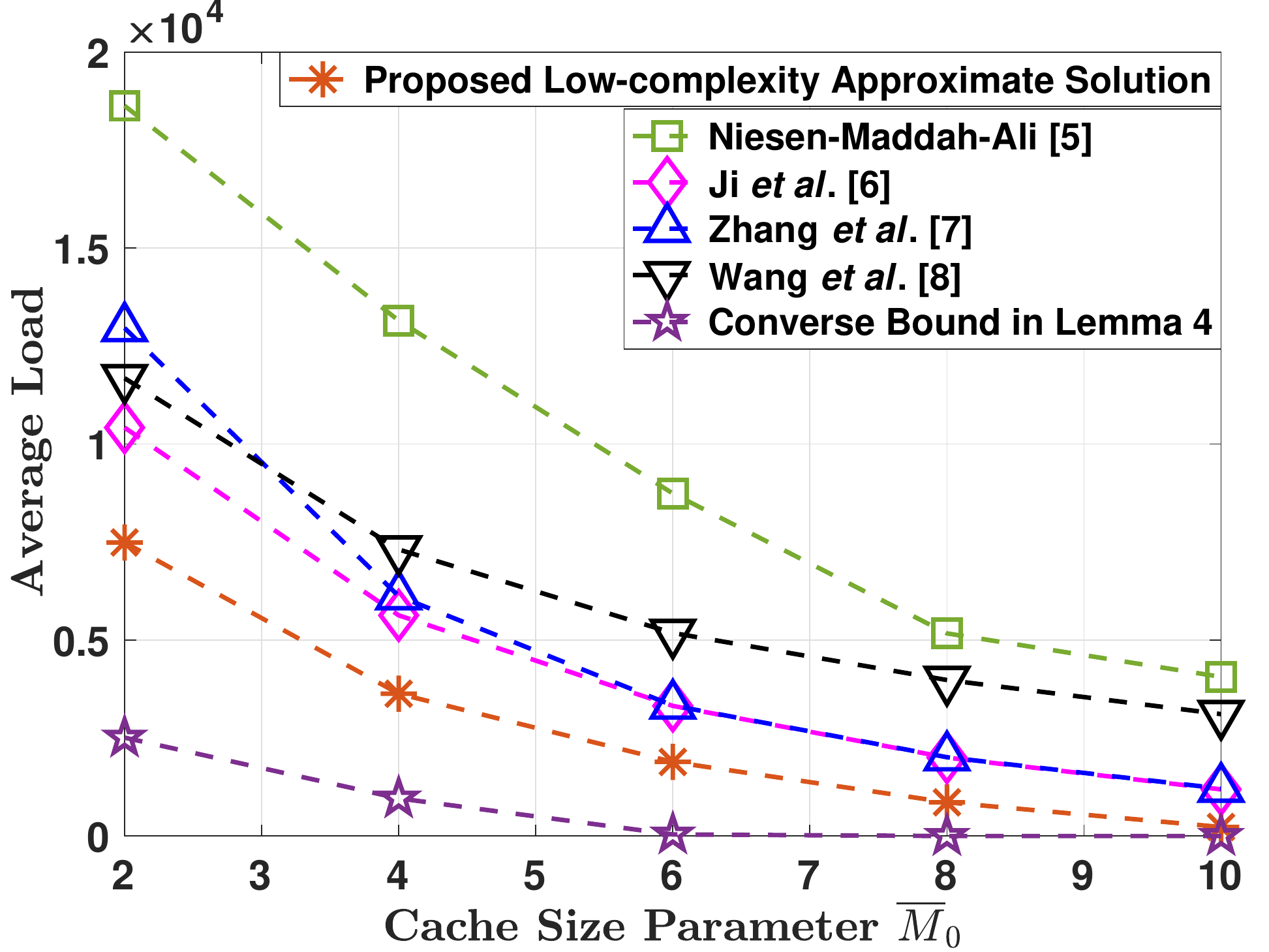}}}
  \end{center}
         \caption{\small{Average load versus cache size parameter $\overline{M}_0$.}
         \vspace{-8mm}
         }\label{fig:simulation_avg4}
\end{figure}
\subsection{Comparison of Average Loads}\label{Subsec:Numerical_Comparison_avg}
In this part, we compare \wqb{the average loads of} the proposed stationary points and low-complexity \wqb{approximate} solution \wqb{in Section~\ref{Sec:Average_Load_Minimization}} with the existing solutions in~\cite{NonuniformDemands,ji2015order,zhang2018coded,Sinong} and the converse bound in Lemma~\ref{L:converse_bound_avg}.
Fig.~\ref{fig:simulation_avg1} (a) and Fig.~\ref{fig:simulation_avg1} (b) illustrate the average loads of the six schemes versus $N$ and \wqb{$T$ (which is the same as $K$ in Fig.~\ref{fig:simulation_avg1})}, respectively.
Fig.~\ref{fig:simulation_avg2} (a) and Fig.~\ref{fig:simulation_avg2} (b) illustrate the average loads of the proposed low-complexity \wqb{approximate} solution and the four baseline schemes versus $\Delta V$ and \wqb{$\Delta \overline{M}$}, respectively.
Fig.~\ref{fig:simulation_avg4} (a) and Fig.~\ref{fig:simulation_avg4} (b) illustrate the average loads of the proposed low-complexity \wqb{approximate} solution \wqb{and} the four baseline schemes, and the converse bound \wqb{in Lemma~\ref{L:converse_bound_avg}} versus cache size parameter \wqb{$\overline{M}_0$}, which determines \wqb{$\left(\overline{M}_t\right)_{t\in\mathcal{T}}$} as illustrated in the caption of Fig.~\ref{fig:simulation_avg4}.
\wqb{Similar observations can be made from Fig.~\ref{fig:simulation1}, Fig.~\ref{fig:simulation2} and Fig.~\ref{fig:simulation3}.
Note that the average loads of \wqb{the existing solutions in\wqm{\cite{NonuniformDemands,ji2015order,zhang2018coded}}} do not change \wqb{smoothly} with some \wqb{system} parameters, as their design parameters are heuristically chosen.}
Fig.~\ref{fig:simulation_avg3} illustrates the average loads of the proposed low-complexity \wqb{approximate} solution and the four baseline schemes versus $\gamma$.
From Fig.~\ref{fig:simulation_avg3}, we can see that among all values of $\gamma$ considered in the simulation, the average load of the proposed low-complexity \wqb{approximate} solution is \wqb{smaller} than those of the decentralized coded caching schemes in~\cite{NonuniformDemands,ji2015order,zhang2018coded,Sinong}, \wqb{indicating that it well adapts to the changes of file popularity.}


From Fig.~\ref{fig:simulation_avg1}, Fig.~\ref{fig:simulation_avg2}, Fig.~\ref{fig:simulation_avg3} and Fig.~\ref{fig:simulation_avg4}, we can see that, the two proposed solutions outperform the four baseline schemes in~\cite{NonuniformDemands,ji2015order,zhang2018coded,Sinong} at the system parameters considered in the simulation, \wqb{owing to their adaptations to  file sizes and cache sizes.}
\section{Conclusion}\label{Sec:Conclusion}
In this paper, we proposed an optimization framework for decentralized coded caching in the general scenario \wqb{with arbitrary file sizes and cache sizes} to minimize the worst-case load \wqb{and} average load.
Specifically, we first proposed a class of decentralized coded caching schemes which are specified by a general caching parameter and include several known schemes as special cases.
Then, we formulated two coded caching design optimization problems over the considered class of schemes to minimize the worst-case load and average load, respectively, with respect to the caching parameter.
\wqb{For each challenging nonconvex optimization problem,} we \wqb{developed} an iterative algorithm to obtain a stationary point and \wqb{proposed} a low-complexity \wqb{approximate} solution with performance guarantee. 
In addition, we presented two information-theoretic converse bounds on the worst-case load and average load (under an arbitrary file popularity) in the general scenario, respectively.
To the best of our knowledge, this is the first work \wqb{that provides optimization-based decentralized coded caching \wqb{schemes} and information-theoretic converse bounds in the general scenario.}
Finally, numerical results showed that the proposed solutions outperform the existing schemes in the general scenario \wqb{and highlighted the importance of designing optimization-based decentralized coded caching schemes for the general scenario}. 
\section*{Appendix A: Proof  of Lemma~\ref{L:lem1}}\label{Appendix_thm1}
First, for all $t\in\mathcal{T}$ and $n\in\mathcal{N}$, we introduce an auxiliary variable $x_{t,n}\ge 0$, replace $1-q_{t,n}$ in $R_{\rm wst}(\mathbf{K},N,\mathbf{V},\mathbf{q})$ with $x_{t,n}$, and add the inequality constraint
\begin{equation}
	1-q_{t,n}\le x_{t,n},\quad\forall t\in\mathcal{T},\ n\in\mathcal{N}.\label{ineqn:auxiliary1}
\end{equation}
Next, for all $\mathbf{n}\in \mathcal{N}^K$, $\mathcal{S}\subseteq \mathcal{K}$, we introduce an auxiliary variable $w_{\mathbf{n},\mathcal{S}}\ge 0$, 
replace $$\max_{j\in \mathcal{S}}
\Bigg(\prod_{a\in \mathcal{S}\backslash\{j\}}\prod_{\substack{t:a\in\mathcal{I}_t(\mathbf{K})}}q_{t,n_j}\Bigg)
\Bigg(\prod_{\substack{b\in\mathcal{K}\backslash  (\mathcal{S}\backslash\{j\})}}\prod_{t:b\in\mathcal{I}_t(\mathbf{K})}(1-q_{t,n_j})\Bigg)V_{n_j}$$ with $w_{\mathbf{n},\mathcal{S}}$, and add the inequality constraint
\begin{align}
	\max_{j\in \mathcal{S}}
\Bigg(\prod_{a\in \mathcal{S}\backslash\{j\}}\prod_{\substack{t:a\in\mathcal{I}_t(\mathbf{K})}}q_{t,n_j}\Bigg)
\Bigg(\prod_{\substack{b\in\mathcal{K}\backslash  (\mathcal{S}\backslash\{j\})}}\prod_{t:b\in\mathcal{I}_t(\mathbf{K})}(1-q_{t,n_j})\Bigg)V_{n_j}\le w_{\mathbf{n},\mathcal{S}},
\quad \forall\mathbf{n}\in \mathcal{N}^K,\ \mathcal{S}\subseteq \mathcal{K}.\label{ineqn:auxiliary2}
\end{align}
Finally, we introduce an auxiliary variable $u\ge 0$, replace $\underset{\mathbf{n}\in \mathcal{N}^{K}}\max\underset{s=1}{\overset{K}{\sum}}\underset{\mathcal{S}\subseteq\mathcal{K}:|\mathcal{S}|=s }{\sum} w_{\mathbf{n},\mathcal{S}}$ with $u$, and add the inequality constraint
\begin{equation}
	\max_{\mathbf{n}\in \mathcal{N}^K}\sum_{s=1}^K\sum_{\mathcal{S}\subseteq\mathcal{K}:|\mathcal{S}|=s }w_{\mathbf{n},\mathcal{S}}\le u.\label{ineqn:auxiliary3}
\end{equation}

Thus, we can equivalently convert Problem~\ref{P:Worst Case1} to Problem~\ref{P:Worst Case2}.
Therefore, we complete the proof.
\section*{Appendix B: Proof  of Theorem~\ref{T:thm1}}\label{Appendix_thm1}
We first prove \eqref{ineqn:T1_1}.
For notation simplicity, denote $V_{\mathbf{n},\mathcal{S},j}\triangleq \Bigg(\underset{a\in \mathcal{S}\backslash\{j\}}\prod \underset{\substack{t:a\in\mathcal{I}_t(\mathbf{K})}}\prod q_{t,n_j}\Bigg)$ $\ \times\ $
$\Bigg(\underset{\substack{b\in\mathcal{K}\\ \backslash  (\mathcal{S}\backslash\{j\})}}\prod \underset{t:b\in\mathcal{I}_t(\mathbf{K})}\prod (1-q_{t,n_j})\Bigg)V_{n_j}.$
By \eqref{eqn:worst_load_1}, we have
\begin{align}
&\quad R_{\rm wst}(\mathbf{K},N,\mathbf{V},\mathbf{q})\overset{(a)}\ge \max_{\mathbf{n}\in \mathcal{N}^K}\sum_{s=1}^K\sum_{\mathcal{S}\subseteq\mathcal{K}:|\mathcal{S}|=s}\left(\frac{1}{c}\ln \Bigg(\sum_{j\in \mathcal{S}}e^{cV_{\mathbf{n},\mathcal{S},j}}\Bigg)-\frac{1}{c}\ln s\right)\nonumber\\
	&=\max_{\mathbf{n}\in \mathcal{N}^K}\sum_{s=1}^K\sum_{\mathcal{S}\subseteq\mathcal{K}:|\mathcal{S}|=s}\frac{1}{c}\ln \Bigg(\sum_{j\in \mathcal{S}}e^{cV_{\mathbf{n},\mathcal{S},j}}\Bigg)-	\sum_{s=1}^K\sum_{\mathcal{S}\subseteq\mathcal{K}:|\mathcal{S}|=s}\frac{1}{c}\ln s\nonumber\\
	&\overset{(b)}\ge \frac{1}{c}\ln\sum_{\mathbf{n}\in \mathcal{N}^K}e^{c\underset{s=1}{\overset{K}\sum} \underset{\mathcal{S}\subseteq\mathcal{K}:|\mathcal{S}|=s}\sum \frac{1}{c}\ln \underset{j\in \mathcal{S}}\sum e^{cV_{\mathbf{n},\mathcal{S},j}}}-\frac{1}{c}\ln N^K
	-\sum_{s=1}^K\sum_{\mathcal{S}\subseteq\mathcal{K}:|\mathcal{S}|=s}\frac{1}{c}\ln s\nonumber\\
	&\overset{(c)}=R_{\rm wst}^{\rm ub}(\mathbf{K},N,\mathbf{V},\mathbf{q})-\frac{1}{c}\left(\sum_{s=1}^K\sum_{\mathcal{S}\subseteq\mathcal{K}:|\mathcal{S}|=s}\ln s+K\ln N\right)\nonumber\\
	&=R_{\rm wst}^{\rm ub}(\mathbf{K},N,\mathbf{V},\mathbf{q})-\frac{1}{c}\left(\sum_{s=1}^K\binom{K}{s}\ln s+K\ln N\right)\triangleq R_{\rm wst}^{\rm lb}(\mathbf{K},N,\mathbf{V},\mathbf{q}),\label{ineqn:R_wst_lb}
\end{align}
where (a) and (b) are due to \eqref{ineqn:max_appro_lb}, and (c) is due to \eqref{ineq:relaxmodel_ub}.
Let $\mathbf{q}^*_{\rm wst}$ denote an optimal solution of Problem~\ref{P:Worst Case1}.
We have
\begin{align}
	&\quad L_{\rm wst}(\mathbf{K},N,\mathbf{V})=R_{\rm wst}(\mathbf{K},N,\mathbf{V},\mathbf{q_{\rm wst}^\dag})-R_{\rm wst}(\mathbf{K},N,\mathbf{V},\mathbf{q_{\rm wst}^*})	\nonumber\\
		 &\overset{(d)}\le R_{\rm wst}^{\rm ub}(\mathbf{K},N,\mathbf{V},\mathbf{q_{\rm wst}^\dag})-R_{\rm wst}(\mathbf{K},N,\mathbf{V},\mathbf{q_{\rm wst}^*})	
		 \overset{(e)}\le R_{\rm wst}^{\rm ub}(\mathbf{K},N,\mathbf{V},\mathbf{q_{\rm wst}^*})-R_{\rm wst}(\mathbf{K},N,\mathbf{V},\mathbf{q_{\rm wst}^*})	\nonumber\\
		 &\overset{(f)}\le R_{\rm wst}^{\rm ub}(\mathbf{K},N,\mathbf{V},\mathbf{q_{\rm wst}^*})-R_{\rm wst}^{\rm lb}(\mathbf{K},N,\mathbf{V},\mathbf{q_{\rm wst}^*})	\overset{(g)}=\frac{1}{c}\left(\sum_{i=1}^K \binom{K}{i}\ln i+K\ln N\right),\label{ineqn:L_and_G}
\end{align}
where (d) is due to \eqref{ineq:relaxmodel_ub}, (e) is due to the optimality of $\mathbf{q}_{\rm wst}^\dag$ for Problem 4, and (f) and (g) are due to \eqref{ineqn:R_wst_lb}.

Next, we prove \eqref{ineqn:T1_2}.
Denote $G_{\rm wst}(K,N)\triangleq \frac{1}{c}\left(\sum_{i=1}^K \binom{K}{i}\ln i+K\ln N\right)$.
We derive two upper bounds on $G_{\rm wst}(K,N)$ as follows:
\begin{align}
G_{\rm wst}(K,N) \overset{(h)}\le \frac{1}{c}\left(\sum_{i=1}^K\binom{K}{i}(i-1)+K\ln N\right)\overset{(i)}=\frac{1}{c}\left(K2^{K-1}-2^K+1+K\ln N\right),\label{ineqn:GAP_bound1}
\end{align}
\begin{align}
G_{\rm wst}(K,N) \overset{(j)}\le \frac{1}{c}\left(\sum_{i=1}^K\binom{K}{i}\ln K+K\ln N\right)=\frac{1}{c}\left(\left(2^K-1\right)\ln K+K\ln N\right),\label{ineqn:GAP_bound2}
\end{align}
where (h) is due to $\ln x\le x-1$, (i) is due to $\sum_{i=1}^K i\binom{K}{i}=K2^{K-1}$, and (j) is due to $\ln i\le \ln K$,  $1\le i\le K$.
By \eqref{ineqn:L_and_G}, \eqref{ineqn:GAP_bound1} and~\eqref{ineqn:GAP_bound2}, we have
\begin{align}
  &\quad L_{\rm wst}(\mathbf{K},N,\mathbf{V})\le \min\left\{\frac{1}{c}\left(K2^{K-1}-2^K+1+K\ln N\right), \frac{1}{c}\left(\left(2^K-1\right)\ln K+K\ln N\right) \right\}\nonumber\\
  &=\frac{1}{c}\Bigg \{\min\left\{\left(\frac{K}{2}-1\right)2^K+1,\left(2^K-1\right)\ln K\right\}+K\ln N\Bigg \}=\frac{1}{c}O\left(2^K\ln K+K\ln N\right). \label{ineqn:Lwst}
\end{align}
Thus, $L_{\rm wst}(\mathbf{K},N,\mathbf{V})=\frac{1}{c}O\left(2^K\ln K+K\ln N\right)$.

Therefore, we complete the proof of Theorem~\ref{T:thm1}.

\section*{Appendix C: Proof  of Theorem~\ref{T:thm2}}\label{Appendix_thm2}
We first prove \eqref{ineqn:T2_1}. 
\wqb{Similar to the proof for \eqref{ineqn:R_wst_lb} in Appendix B}, we have
\begin{align}
\quad R_{\rm avg}(\mathbf{K},N,\mathbf{V},\mathbf{q})\ge 
	R_{\rm avg}^{\rm ub}(\mathbf{K},N,\mathbf{V},\mathbf{q})-\frac{1}{c}\sum_{s=1}^K\binom{K}{s}\ln s\triangleq R_{\rm avg}^{\rm lb}(\mathbf{K},N,\mathbf{V},\mathbf{q}).\label{ineqn:R_avg_lb}
\end{align}
Let $\mathbf{q}^*_{\rm avg}$ denote an optimal solution of Problem~\ref{P:Average Case1}.
\wqb{Similar to the proof for \eqref{ineqn:L_and_G} in Appendix B}, we have
\begin{align}
	L_{\rm avg}(\mathbf{K},N,\mathbf{V})
		 &\le R_{\rm avg}^{\rm ub}(\mathbf{K},N,\mathbf{V},\mathbf{q^*_{\rm avg}})-R_{\rm avg}^{\rm lb}(\mathbf{K},N,\mathbf{V},\mathbf{q^*_{\rm avg}})
		 \overset{(a)}=\frac{1}{c}\sum_{i=1}^K \binom{K}{i}\ln i,
\end{align}
where (a) is due to \eqref{ineqn:R_avg_lb}.
Next, we prove \eqref{ineqn:T2_2}.
\wqb{Similar to the proof for \eqref{ineqn:Lwst} in Appendix B, we have}
\begin{align}
  &\quad L_{\rm avg}(\mathbf{K},N,\mathbf{V})\le \min\left\{\frac{1}{c}\left(K2^{K-1}-2^K+1\right), \frac{1}{c}(2^K-1)\ln K \right\}\nonumber\\
  &=\frac{1}{c}\min\left\{\left(\frac{K}{2}-1\right)2^K+1,\left(2^K-1\right)\ln K\right\}=\frac{1}{c}O\left(2^K\ln K\right) \label{ineqn:lemma2}.
\end{align}
Thus, $L_{\rm avg}(\mathbf{K},N,\mathbf{V})=\frac{1}{c}O\left(2^K\ln K\right)$.
Therefore, we complete the proof of Theorem~\ref{T:thm2}.
\section*{Appendix D: Proof  of Lemma~\ref{L:converse_bound_wst}}\label{Appendix_lem_converse_wst}
\wqb{For notation simplicity}, we assume that $\mathcal{L}_a=\{1,2,\dots,L_a\}$ and $M_1\le M_2\le \dots\le M_{L_a}$ (i.e., $M_{[\mathcal{L}_a,i]}=M_i, 1 \le i \le L_a$) \wqb{in the proof}.
\wqb{By the proof of Lemma 1 in~\cite{Chien_Yi_Wang2018improved}, for any $m\in\big\{1,\dots,\min\{L_a,N\}\big\}$, we have}
\begin{align}
R(\mathcal{L}_a,N,\mathbf{V},\mathbf{M},\mathbf{d}_a)&\ge\sum_{l=1}^{m} H(W_{d_l})\cdot\mathbbm{1}\big\{d_l\notin\{d_1,\dots,d_{l-1}\}\big\}\nonumber\\
	&\quad-\sum_{l=1}^m I(W_{d_l};Z_1,\dots,Z_l|W_{d_1},\dots,W_{d_{l-1}}),\label{ineqn:converse_bound_wst_proof1}	
\end{align}
where $\mathbbm{1}\big\{d_l\notin\{d_1,\dots,d_{l-1}\}\big\}$ denotes the indicator function that is 1 if $d_l$ is not in $\{d_1,\dots,d_{l-1}\}$ and is 0 otherwise.
Let $\mathcal{Q}_m^{\rm dist}$ denote the set of all ordered $m$-dimentional demand vectors $(d_1,\dots,d_m)$ with all distinct entries.
Hence, $|\mathcal{Q}_m^{\rm dist}|=\binom{N}{m}m!$.
Averaging \eqref{ineqn:converse_bound_wst_proof1} over all demand vectors $\mathbf{d}\in\mathcal{Q}_m^{\rm dist}$ yields:
\begin{align}
	R_{\rm wst}^{*}(\mathcal{L}_a,N,\mathbf{V},\mathbf{M})\ge\frac{1}{\binom{N}{m}m!}\sum_{\mathbf{d}\in\mathcal{Q}_m^{\rm dist}}\sum_{l=1}^m H(W_{d_l})\cdot\mathbbm{1}\big\{d_l\notin\{d_1,\dots,d_{l-1}\}\big\}-\sum_{l=1}^m\alpha_l,\label{ineqn:*}
\end{align}
where $\alpha_1\triangleq\frac{1}{\binom{N}{m}m!}\sum_{\mathbf{d}\in\mathcal{Q}_m^{\rm dist}}I(W_{d_1};Z_1)$, \wqb{and} $\alpha_l\triangleq\frac{1}{\binom{N}{m}m!}\sum_{\mathbf{d}\in\mathcal{Q}_m^{\rm dist}}I(W_{d_l};Z_1,\dots,Z_l|W_{d_1},\dots,W_{d_{l-1}})$, $l=2,\dots,m$.
\wqb{In the following, we derive an lower bound of the lower bound in \eqref{ineqn:*} by analyzing its two terms, respectively.}
\wqb{First, we} simplify the first term of the lower bound in \eqref{ineqn:*} as follows:
	\begin{align}
		&\quad\frac{1}{\binom{N}{m}m!}\sum_{\mathbf{d}\in\mathcal{Q}_m^{\rm dist}}\sum_{l=1}^m H(W_{d_l})\cdot\mathbbm{1}\big\{d_l\notin\{d_1,\dots,d_{l-1}\}\big\}
		=\frac{1}{\binom{N}{m}m!}m!\sum_{\mathcal{S}\subseteq\mathcal{N}:|\mathcal{S}|=m}\sum_{i\in\mathcal{S}}V_i\nonumber\\
		&=\frac{1}{\binom{N}{m}m!}m!\binom{N-1}{m-1}\sum_{i=1}^{N}V_i
		=\frac{m}{N}\sum_{i=1}^N V_i.\label{ineqn:*_1}
	\end{align}
\wqb{Next}, we derive an upper bound of the second term of the lower bound in \eqref{ineqn:*}.
\wqb{For all $1\le l \le m$, let $\tilde{\mathbf{d}}\triangleq (d_1,\dots,d_{l-1})$ and $W_{\tilde{\mathbf{d}}}\triangleq\{W_{d_1},\dots,W_{d_{l-1}}\}$.
By the definition of $\alpha_l,\ 2\le l \le m$, we have}
\begin{align}
	\alpha_l
		    &\overset{(a)}{\le}\frac{1}{l!\binom{N}{l}}\sum_{\tilde{\mathbf{d}}\in{\mathcal{Q}_{l-1}^{\rm dist}}}I\big(\{W_j:j\in\mathcal{N}\backslash\tilde{\mathbf{d}} \};Z_1,\dots,Z_l|W_{\tilde{\mathbf{d}}}\big)
		    \overset{(b)}{\le}\frac{1}{l!\binom{N}{l}}\sum_{\tilde{\mathbf{d}}\in{\mathcal{Q}_{l-1}^{\rm dist}}}\sum_{i=1}^l M_{i}
		    =\frac{\sum_{i=1}^lM_{i}}{N-l+1},\nonumber
\end{align}
where (a) is due to (88) \wqb{in the proof of Lemma 3} in~\cite{Chien_Yi_Wang2018improved} and (b) is due to \wqb{$I\big(\{W_j:j\in\mathcal{N}\backslash\tilde{\mathbf{d}} \};Z_1,\dots,Z_l|W_{\tilde{\mathbf{d}}}\big)\le \sum_{i=1}^l H(Z_i)$, $\tilde{\mathbf{d}}\in{\mathcal{Q}_{l-1}^{\rm dist}}$ and} $H(Z_i)\le M_{i}$, $1 \le i \le l$.
\wqb{Similarly, we have $\alpha_1\le \frac{M_1}{N}$.
Thus, we have}  
\begin{align}
	\sum_{l=1}^m \alpha_l\le\sum_{l=1}^m \frac{\sum_{i=1}^{l}M_{i}}{N-l+1}.\label{ineqn:sum_alpha1}
\end{align}
\wqb{In addition, by the definition of $\alpha_l,\ 1\le l\le m$, we have}
\begin{align}
	&\binom{N}{m}m!\sum_{l=1}^m\alpha_l
	\overset{(c)}{\le} m!\binom{N}{m}\frac{m}{N}
	I(W_1,\dots,W_N;Z_1,\dots,Z_m)
	\overset{(d)}{\le} m!\binom{N}{m}\frac{m}{N}\sum_{i=1}^m M_{i},\label{ineqn:sum_alpha2}
\end{align}
where (c) is due to (91) \wqb{in the proof of Lemma 3} in~\cite{Chien_Yi_Wang2018improved} and (d) is due to \wqb{$I(W_1,\dots,W_N;Z_1,\dots,Z_m)$ $\le \sum_{i=1}^m H(Z_i)$} and $H(Z_i)\le M_{i}$, $1 \le i \le m$.
By \eqref{ineqn:sum_alpha1} and \eqref{ineqn:sum_alpha2}, we have
\begin{align}
	\sum_{l=1}^m\alpha_l\le\min\left\{\sum_{l=1}^m\frac{\sum_{i=1}^l M_{i}}{N-l+1},\frac{m}{N}\sum_{i=1}^m M_{i}\right\}	.\label{ineqn:converse_bound_sum_alpha}
\end{align}
\wqb{Finally}, by \eqref{ineqn:*}, \eqref{ineqn:*_1}, \eqref{ineqn:converse_bound_sum_alpha} and optimizing over all possible choices of $m\in\big\{1,\dots,\min\{L_a,N\}\big\}$, we \wqb{can show} Lemma \ref{L:converse_bound_wst}.
\section*{Appendix E: Proof  of Lemma~\ref{L:converse_bound_avg}}\label{Appendix_lem_converse_avg}
Let $R_{\rm avg,unif}^*(\mathcal{S},N^{'},\mathbf{V},\mathbf{M})$ denote the minimum average load of the shared link \wqb{for the} active users in $\mathcal{S}$ when the library size is $N^{'}$, file sizes are $(V_n)_{n=1}^{N^{'}}$ and cache sizes are $\mathbf{M}$ \wqb{and the file popularity distribution follows the uniform distribution}. 
\wqb{Following the proof of Theorem 2 in\cite{ji2015order}}, we have
\begin{align}
	R_{\rm avg}^{*}&(\mathcal{L}_a,N,\mathbf{V},\mathbf{M})
		 \ge \max_{N^{'}\in\{1,\dots,N\}}\left\{\sum_{i=1}^{L_a}\sum_{\mathcal{S}\subseteq\mathcal{L}_a:|\mathcal{S}|=i}\big(N^{'}p_{N^{'}}\big)^i\big(1-N^{'}p_{N^{'}}\big)^{L_a-i} R_{\rm avg,unif}^*(\mathcal{S},N^{'},\mathbf{V},\mathbf{M})\right\}\nonumber\\
		 &=\max_{N^{'}\in\{1,\dots,N\}}\left\{\sum_{i=1}^{L_a} \left(N^{'}p_{N^{'}}\right)^i \left(1-N^{'}p_{N^{'}}\right)^{L_a-i}
		 \sum_{\substack{\mathcal{S}\subseteq\mathcal{L}_a: |\mathcal{S}|=i}} R_{\rm avg,unif}^{*}(\mathcal{S},N^{'},\mathbf{V},\mathbf{M})\right\}.\label{ineqn:appendix_genie-aided}
\end{align}

Next, we \wqb{derive a lower bound on $R_{\rm avg,unif}^{*}(\mathcal{L}_a,N,\mathbf{V},\mathbf{M})$}.
\wqb{Without loss of generality, we assume that $\mathcal{L}_a=\{1,2,\dots,L_a\}$ and $M_1\le M_2\le \dots\le M_{L_a}$ (i.e., $M_{[\mathcal{L}_a,i]}=M_i, 1 \le i \le L_a$).}
Fix $m\in\big\{1,\dots,\min\{L_a,N\}\big\}$.
Let $\mathcal{Q}_m^{\rm rep}$ denote the set of all ordered $m$-dimensional demand vectors, where repetitions are allowed.
Hence, $|\mathcal{Q}_m^{\rm rep}|=N^m$.
\wqb{Assume that the random requests of all $L_a$ users, i.e., $\mathbf{D}_a\triangleq (D_1,\dots,D_{L_a})$, follow the uniform distribution.}
Averaging \eqref{ineqn:converse_bound_wst_proof1} over all demand vectors $\mathbf{d}\in\mathcal{Q}_{L_a}^{\rm rep}$ yields:
\begin{align}
	R_{\rm avg,unif}^{*}(\mathcal{L}_a,N,\mathbf{V},\mathbf{M})
	\ge\mathbb{E}\Bigg[\sum_{l=1}^m H(W_{D_l})\cdot\mathbbm{1}\big\{D_l\notin\{D_1,\dots,D_{l-1}\}\big\}\Bigg]-\sum_{l=1}^m\beta_l,\label{ineqn:*_avg}
\end{align}
where \wqb{the expectation is with respect to $\mathbf{D}_a$}, $\beta_1\triangleq I(W_{D_1};Z_1|\mathbf{D}_a)$ and $\beta_l\triangleq I(W_{D_l};Z_1,\dots,Z_l|W_{D_1},$ $\dots,W_{D_{l-1}},\mathbf{D}_a)$, $l=2,\dots,m$.
\wqb{In the following, we derive a lower bound of the lower bound on \eqref{ineqn:*_avg} by analyzing its two terms, respectively.}
\wqb{First, we} simplify the first term of the lower bound derived in \eqref{ineqn:*_avg} as follows:
	\begin{align}
		&\quad\mathbb{E}\Bigg[\sum_{l=1}^m H(W_{D_l})\cdot\mathbbm{1}\big\{D_l\notin\{D_1,\dots,D_{l-1}\}\big\}\Bigg]\nonumber\\
		&=\sum_{j=1}^m\sum_{\mathcal{S}\subseteq\mathcal{N}:|\mathcal{S}|=j}\Pr\Bigg(\sum_{l=1}^l H(W_{D_l})\cdot\mathbbm{1}\big\{D_l\notin\{D_1,\dots,D_{l-1}\}\big\}
		 =\sum_{i\in\mathcal{S}} V_i \Bigg)\sum_{i\in\mathcal{S}}V_i\nonumber\\
		&=\sum_{j=1}^m\sum_{\mathcal{S}\subseteq\mathcal{N}:|\mathcal{S}|=j}\frac{j!\stirling{m}{j}}{N^m}\sum_{i\in\mathcal{S}}V_i
		=\sum_{j=1}^m\frac{j!\stirling{m}{j}}{N^m}\binom{N-1}{j-1}\sum_{i=1}^N V_i,\label{ineqn:*_2}
	\end{align}
\wqb{Next,} we derive an upper bound on the second term of the lower bound derived in \eqref{ineqn:*_avg}.
\wqb{By the definition of $\beta_l,\ 2\le l \le m$, we have}
\begin{align}
	&\beta_l
		    \overset{(a)}{\le}\frac{1}{N^l}\sum_{\tilde{\mathbf{d}}\in{\mathcal{Q}_{l-1}^{\rm rep}}}I\big(W_1,\dots,W_N;Z_1,\dots,Z_l|W_{\tilde{\mathbf{d}}}\big)
		    \overset{(b)}{\le} \frac{1}{N^l}\sum_{\tilde{\mathbf{d}}\in{\mathcal{Q}_{l-1}^{\rm rep}}}\sum_{i=1}^l M_{i}
		    =\frac{\sum_{i=1}^l M_{i}}{N},\nonumber
\end{align}
where (a) is due to (94) in the proof of Lemma 5 in~\cite{Chien_Yi_Wang2018improved} and (b) is due to \wqb{$I\big(W_1,\dots,W_N;Z_1,\dots,Z_l|W_{\tilde{\mathbf{d}}}\big)$ $\le \sum_{i=1}^l H(Z_l)$, $\tilde{\mathbf{d}}\in{\mathcal{Q}_{l-1}^{\rm rep}}$} and $H(Z_i)\le M_{i}$, $1 \le i \le l$.
\wqb{Similarly, we have $\beta_1\le\frac{M_1}{N}$}.
\wqb{Thus}, we have
\begin{align}
	\sum_{l=1}^m \beta_l\le\sum_{l=1}^m \frac{\sum_{i=1}^{l}M_{i}}{N}.\label{ineqn:sum_beta1}
\end{align}
\wqb{In addition, by the definition of $\beta_l,\ 1\le l\le m$,}
\begin{align}
	\sum_{l=1}^m\beta_l
	&\overset{(c)}{\le}\sum_{i=1}^m \Pr\big(\kappa_{\mathbf{D}_a}(m)=i\big)\frac{i}{N}I\big(W_1,\dots,W_N;Z_1,\dots,Z_m\big)\nonumber\\
	&\overset{(d)}{\le}  \sum_{i=1}^m \Pr \big(\kappa_{\mathbf{D}_a}(m)=i\big)\cdot i\cdot\frac{\sum_{i=1}^m M_{i} }{N}
	=\mathbb{E}\big[\kappa_{\mathbf{D}_a}(m)\big] \frac{\sum_{i=1}^m M_{i}}{N}, \label{ineqn:sum_beta2}
\end{align}
where  (c) is due to (96) \wqb{in the proof of Lemma 5}  in~\cite{Chien_Yi_Wang2018improved} and (d) is due to \wqb{$I\big(W_1,\dots,W_N;Z_1,\dots,Z_m\big)$ $\le \sum_{i=1}^m H(Z_i)$} and $H(Z_i)\le M_{i}$, $1 \le i \le m$.
\wqb{Here $\kappa_{\mathbf{D}_a}(m)$ denotes the number of distinct demands for users $1,\dots,m$.}
By \eqref{ineqn:sum_beta1} and \eqref{ineqn:sum_beta2}, we have
	\begin{align}
	\sum_{l=1}^m \beta_l \le\min\left\{\sum_{l=1}^m\frac{\sum_{i=1}^l M_{i}}{N},\mathbb{E}\big[\kappa_{\mathbf{D}_a}(m)\big] \frac{\sum_{i=1}^m M_{i}}{N}\right\}.\label{ineqn:converse_bound_sum_beta}
	\end{align}
\wqb{By Lemma 4 in~\cite{Chien_Yi_Wang2018improved}, we have}
	\begin{align}
		\mathbb{E}\big[\kappa_{\mathbf{D}_a}(m)\big]=N\Big(1-\Big(1-\frac{1}{N}\Big)^m\Big)	.\label{eqn:converse_bound_avg_proof_E_kappa}
	\end{align}			
\wqb{Finally,} by \eqref{ineqn:*_avg}, \eqref{ineqn:*_2}, \eqref{ineqn:converse_bound_sum_beta}, \eqref{eqn:converse_bound_avg_proof_E_kappa} and optimizing over all possible choices of $m\in\big\{1,\dots,\min\{L_a,N\}\big\}$, we have
\begin{align}
R_{\rm avg,unif}^{*}(\mathcal{L}_a,N,\mathbf{V},\mathbf{M})\ge R_{\rm avg,unif}^{\rm lb}(\mathcal{L}_a,N,\mathbf{V},\mathbf{M}).	\label{ineqn:converse_bound_unif_uniflb}
\end{align}
Therefore, by \eqref{ineqn:appendix_genie-aided} and \eqref{ineqn:converse_bound_unif_uniflb}, we \wqb{can show} Lemma 4.
\bibliography{myReferences}
\end{document}